\pdfoutput=1
\documentclass[11pt,letterpaper]{article}
\usepackage[T1]{fontenc}
\usepackage[utf8]{inputenc}
\usepackage[USenglish]{babel}
\usepackage{lmodern}
\usepackage[margin=1in]{geometry}
\usepackage{microtype}
\usepackage{amsmath,amssymb,amsthm,mathtools}
\usepackage{booktabs,tabularx,array}
\usepackage[ruled]{algorithm}
\usepackage[noend]{algpseudocode}
\usepackage{xcolor}
\usepackage{url}
\usepackage[section]{placeins}
\usepackage[unicode,hidelinks]{hyperref}
\usepackage{bookmark}
\newtheorem{theorem}{Theorem}[section]
\newtheorem{lemma}[theorem]{Lemma}
\newtheorem{corollary}[theorem]{Corollary}
\newtheorem{proposition}[theorem]{Proposition}
\theoremstyle{remark}
\newtheorem{remark}[theorem]{Remark}
\usepackage{tikz}
\usetikzlibrary{arrows.meta,positioning,calc,fit,decorations.pathreplacing}
\newcolumntype{Y}{>{\raggedright\arraybackslash}X}
\newcommand{\E}{\mathbb E}
\newcommand{\Prb}{\mathbb P}

\newcommand{\cF}{\mathcal F}
\newcommand{\cost}{\operatorname{cost}}
\newcommand{\OPT}{\operatorname{OPT}}
\newcommand{\EmpiricalSort}{\textnormal{\textsc{EmpiricalSort}}}
\newcommand{\proofpointer}[2]{%
  \par\nopagebreak[4]
  \noindent\emph{\hyperref[#2]{Proof in Section~\ref*{#1}.}}\par
}

\DeclareMathOperator{\capc}{cap}

\renewenvironment{abstract}{%
  \begin{center}\bfseries Abstract\end{center}\quotation\normalsize
}{\endquotation}

\makeatletter
\renewcommand{\@maketitle}{%
  \newpage\null\vskip0.75em
  \begin{center}
    {\LARGE\@title\par}
    \ifx\@author\@empty\else
      \vskip1.2em
      {\large\lineskip.5em\begin{tabular}[t]{c}\@author\end{tabular}\par}
    \fi
  \end{center}\par\vskip0.5em
}

\renewcommand*\l@section[2]{%
  \ifnum\c@tocdepth>\z@
    \addpenalty{-\@highpenalty}%
    \vskip0.30em \@plus0.5pt
    \begingroup\bfseries
    \@dottedtocline{1}{0em}{1.7em}{#1}{#2}%
    \endgroup
  \fi}

\makeatother
\title{Optimal High Probability Bounds for Random Order Online Sorting}
\author{
  Chenglin Fan\thanks{Seoul National University.}
  \and
  Mingyu Kang\footnotemark[1]
}
\date{}
\hypersetup{pdftitle={Optimal High Probability Bounds for Random-Order Online Sorting},
  pdfsubject={Exact-space random-order online sorting},
  pdfkeywords={Online sorting, random-order algorithms, competitive analysis, concentration inequalities}}
\begin{document}
\maketitle
\begin{abstract}
In random order online sorting, a fixed multiset of $n$ real numbers arrives
in a uniformly random order. Each value must be placed immediately and
irrevocably into an empty position of an array with exactly $n$ positions.
The objective is to minimize the sum of the absolute differences between
consecutive entries. If the entire input were known in advance, sorting the
values would be optimal and would have cost equal to the range of the input.

We prove that the optimal worst case competitive ratio in this model is
$\Theta(\log n)$ with high probability. In particular, we give a single
online algorithm that achieves an $O(\log n)$ competitive ratio with high probability for every
fixed input, and we show that every randomized online algorithm has some
fixed input on which an $\Omega(\log n)$ ratio is unavoidable with high
probability. The same upper bound also holds for independent samples from
an arbitrary unknown distribution on the real line.

The algorithm uses the random arrival order to learn the relative positions
of future values from the values seen so far. It maintains a hierarchy of
small regions of the input order and allows an arriving value to use nearby
space when its preferred region is full. This flexibility avoids the need
to reserve large amounts of space for unlikely local fluctuations. The
analysis shows that large displacements become rapidly less likely as their
scale grows, leading to the $O(\log n)$ high probability bound. A matching
lower bound shows that no randomized online algorithm can improve the
logarithmic dependence in the worst case.
\end{abstract}
\clearpage
\tableofcontents
\clearpage


\section{Introduction}\label{sec:introduction}

An online sorting algorithm receives $n$ real numbers one at a time and places
each of them immediately into an empty position of an array with exactly $n$
positions. Once placed, a value cannot be moved. The goal is to minimize
\[
    \cost(A)=\sum_{i=1}^{n-1}|A[i+1]-A[i]|,
\]
where $A$ is the completed array. With the entire input known in advance,
sorting attains the minimum possible cost, equal to the input range.

We study a fixed multiset presented in a uniformly random order. Its values
may repeat and need not be evenly spaced. An observed prefix gives information
about the ranks of future arrivals, but the values in that prefix already
occupy positions in the array. The algorithm must use what it has learned to
decide how to assign the remaining positions.
We give a single algorithm using only comparisons that achieves an $O(\log n)$
competitive ratio with high probability for every fixed multiset. The ratio
has an exponential tail above the logarithmic scale. We also prove that every
randomized online algorithm has a fixed input on which its ratio is
$\Omega(\log n)$ with probability $1-e^{-\Omega(\sqrt n)}$. The optimal ratio in the worst
case is therefore $\Theta(\log n)$ with high probability. The upper bound
also applies to independent samples from an arbitrary unknown distribution
on the real line.

Our algorithm uses the observed values to partition the value space into
intervals. Each interval is assigned a bucket consisting of a constant
number of empty array cells at the start of a phase; these cells need not
be consecutive. The buckets are organized in a binary tree. When the
bucket for an arrival's interval is full, the tree determines where to
find available space. Both the partitions and the empty cells left for
later phases depend on earlier arrivals. The analysis must account for
this dependence when bounding the cost of overflow.

\subsection{Model and main results}\label{sec:model}

Let $X=\{x_1,\ldots,x_n\}$ be a fixed multiset of $n\ge2$ finite real values,
with indices in $[n]=\{1,\ldots,n\}$ distinguishing equal occurrences. A
uniform permutation $\pi$ of $[n]$ gives the stream $X_\pi=(x_{\pi(1)},\ldots,x_{\pi(n)})$. The algorithm
knows $n$ and must place each arrival into an empty array cell before seeing
the next one. It has exactly $n$ output cells, and filled cells never move.
For the completed array $A$, define
\begin{equation}\label{eq:ratio-definition}
 \begin{aligned}
 \cost(A)&=\sum_{i=1}^{n-1}|A[i+1]-A[i]|,
 \qquad \OPT(X)=\max X-\min X,\\
 \rho(A)&=\begin{cases}
 \cost(A)/\OPT(X),&\OPT(X)>0,\\
 1,&\OPT(X)=0.
 \end{cases}
 \end{aligned}
\end{equation}
Every ordering must connect a minimum value to a maximum value, whereas
sorted order attains exactly this range. Thus $\OPT(X)$ is the offline
optimum and $\rho(A)$ is the realized competitive ratio. We always have
$1\le\rho(A)\le n-1$; when all values are equal, both costs are zero and the
ratio is one by convention. We use $\log$ for the natural logarithm and
$\log_2$ for the binary logarithm, and let $\log^*n$
count the applications of $\log_2$ needed to reach a value at most $1$.

For equal values, our algorithm uses an independent uniform permutation
$\sigma$ of $[n]$ to break ties. The arrival at time $t$ is compared using the
lexicographic key $(x_{\pi(t)},\sigma(t))$, but the output stores only its
original value. These keys are strictly ordered, and the priorities have no
effect when the input values are distinct. Section~\ref{tr:metric-tail}
justifies this construction. We count exact comparisons between real values
and word operations on counters with $O(\log n)$ bits. Auxiliary memory is
allowed; the restriction to $n$ cells concerns the output array.

\begin{samepage}
\begin{theorem}[Upper tail in random order]\label{thm:upper}
There are absolute constants $a_0,a_1,\theta>0$ and a single online algorithm
using only comparisons such that, for every $n\ge2$, every fixed multiset $X$,
and every $z\ge0$, its output $A$ satisfies
\begin{equation}\label{eq:main-upper-tail}
 \Prb_{\pi,\sigma}\!\left[\rho(A)>a_0\log n+a_1+z\right]
 \le e^{-\theta z}.
\end{equation}
The algorithm completes a legal placement on every execution. It is
deterministic on distinct inputs and uses the independent priority
permutation $\sigma$ to break ties. After generating the priorities, it
requires $O(n\log n)$ comparisons and word operations and $O(n)$ auxiliary
words.
\end{theorem}
\proofpointer{sec:transfer}{proof:upper}
\end{samepage}

The algorithm is \EmpiricalSort{}, described in Section~\ref{sec:overview}.
For any fixed $d>0$, taking $z=d\log n/\theta$ in
\eqref{eq:main-upper-tail} gives
$\Prb[\rho>K_d\log n]\le n^{-d}$, where
\[
    K_d=a_0+\frac d\theta+\frac{a_1}{\log 2}.
\]
The algorithm does not depend on $d$. The guarantee also extends to
samples from an unknown distribution.

\begin{samepage}
\begin{corollary}[Unknown i.i.d.\ input]\label{cor:unknown-iid}
For every fixed $d>0$, the same algorithm satisfies
$\sup_{F\in\mathcal P(\mathbb R)}
\Prb_{F^{\otimes n},\sigma}[\rho(A)>K_d\log n]\le n^{-d}$,
where $\mathcal P(\mathbb R)$ is the set of Borel probability laws on
$\mathbb R$. The algorithm receives no information about $F$.
\end{corollary}
\proofpointer{tr:iid}{proof:unknown-iid}
\end{samepage}

Exchangeability gives the corollary by conditioning on the sampled multiset.
Integrating the tail also gives $\E[\rho]=O(\log n)$ uniformly over $F$.
No continuity or moment assumptions are needed: the ratio is bounded by
$n-1$, even when the input values themselves have no finite moments.

\paragraph*{Matching lower bounds.}
The lower bounds allow unrestricted auxiliary storage and access to the full
numerical values, rather than just comparisons.

\begin{samepage}
\begin{theorem}[Lower tail for uniform inputs]\label{thm:uniform-lower}
There are absolute constants $c_0,c_1>0$ and $n_0$ such that, for every
$n\ge n_0$, every possibly randomized measurable online algorithm $\mathcal A$
using exactly $n$ cells satisfies
$\Prb[\cost(\mathcal A(X_1,\ldots,X_n))<c_0\log n]
\le4e^{-c_1\sqrt n}$
on independent $U[0,1]$ inputs. Its internal random tape $\zeta$ is
independent of the input.
\end{theorem}
\proofpointer{lb:section}{proof:uniform-lower}
\end{samepage}

\begin{samepage}
\begin{theorem}[Lower tail in random order]\label{thm:random-lower}
For every $n\ge n_0$ and every such online algorithm $\mathcal A$, there is a
fixed set $X\subset[0,1]$ of $n$ distinct values such that
$\Prb_{\pi,\zeta}[\rho(\mathcal A(X_\pi))<c_0\log n]
\le4e^{-c_1\sqrt n}$.
The set may depend on $\mathcal A$ and $n$, but is fixed before the
permutation and the independent algorithmic tape are drawn.
\end{theorem}
\proofpointer{lb:section}{proof:random-lower}
\end{samepage}

\paragraph*{Optimality.}
For uniform inputs, $\OPT(X)\le1$, so Theorem~\ref{thm:upper},
Corollary~\ref{cor:unknown-iid}, and Theorem~\ref{thm:uniform-lower} give
matching $\Theta(\log n)$ bounds on cost with high probability.
Theorems~\ref{thm:upper} and~\ref{thm:random-lower} give the same order for
the competitive ratio in the worst case over fixed inputs arriving in
random order.

\paragraph*{Bounds for structured inputs.}
The ratio can be smaller when the large gaps in the sorted input occur near
its ends. Section~\ref{sec:instance-sensitive} derives such bounds from the
same analysis. For every fixed $d>0$, suppose a nonconstant input has at
least $n-r$ values in an interval of width at most $\varepsilon\OPT(X)$,
where $r\in\{0,\ldots,n-1\}$ and $0\le\varepsilon\le1$. Then its ratio is
$O_d(r+\varepsilon\log n)$ with failure probability at most $n^{-d}$,
where the implicit constant depends only on $d$.
For $\{1,1/2,\ldots,1/n\}$, the ratio is $O_d(1+\log\log n)$ with the
same failure probability. Neither guarantee requires the algorithm to
know these properties of the input.

\subsection{Previous work and our approach}\label{sec:related}

Partitioning the value space and assigning capacity to each interval leaves
an immediate difficulty: an interval may receive more arrivals than its
capacity allows. Earlier algorithms handle this overflow in different ways.

For independent uniform inputs, Hu~\cite{Hu2026} gave an algorithm with
expected cost $\log n\,2^{O(\log^* n)}$, using adaptive buffers,
synchronization, and recursion. Hu also proved an $\Omega(\log n)$ lower
bound in expectation. Kalavas et~al.~\cite{KPT2026} use a sequence of
balls-into-bins stages, recycling unused capacity between stages. Bins large
enough for simultaneous concentration give an $O(\log^2 n)$ bound with high
probability; they also obtain a refinement to
$O(\log^{3/2+\delta}n)$ for every fixed $\delta>0$.
Hermansen~\cite{Hermansen2026} achieved expected cost $O(\log n)$ with
buckets of constant capacity arranged in a binary hierarchy. A request whose
bucket is full uses nearby capacity in the tree. The analysis bounds the
cost of this routing instead of requiring all buckets to avoid overflow.
The partition is fixed using the known uniform distribution.

For arbitrary fixed inputs arriving in random order, a numerical interval
can contain any number of values, regardless of its length. Platanos and
Tolias~\cite{PlatanosTolias2026} estimate quantiles from random prefixes and
construct buckets from these estimates. They obtain an $O(\log^2 n)$
competitive ratio with high probability and announce an
$O(\log^{3/2}n)$ refinement.

We use empirical partitions with buckets of constant capacity and
Hermansen's routing rule. At first, larger prefixes define finer partitions.
After the final partition is built, the algorithm merges neighboring
intervals and reuses the empty cells. At each reallocation, only empty
cells are reassigned; previously placed values stay where they are.
For the upper bound, we first analyze independent uniform inputs on $[0,1]$.
We count adjacent output pairs that lie on opposite sides of a fixed
threshold. During growth, the observed prefix doubles in size from one
phase to the next. We expose these prefixes in reverse, revealing only
their unordered sets of values. Conditional on the larger set, the
preceding prefix is a uniformly chosen subset of half its size. The
resulting conditional bounds can be combined across phases by successive
conditioning.
After the final partition is fixed, we separate its errors from fluctuations
in the later arrivals. An inaccurate interval covering $\nu$ of the finest
buckets incurs geometrically decreasing charges as buckets are merged, for
a total of $O(\nu)$ rather than $O(\nu\log n)$. Reallocation demand is bounded by
counts from the preceding input block, so the bound does not require the
distribution of the surviving empty cells. These estimates give an
exponential tail for the crossing count above the logarithmic scale.
Because placement decisions depend only on the relative order of the
inputs, we can transfer the crossing estimates to cuts between consecutive
ranks of any fixed input. Weighting each crossing count by the
corresponding numerical gap recovers the total cost.

For the lower bound, run an arbitrary algorithm on independent uniform
inputs on $[0,1]$ and stop with $i$ arrivals remaining. At most $2i$ occupied
cells border an empty position; call these cells anchors. For a sufficiently
small absolute constant $c>0$, connect two remaining values when their
distance is at most $c/i$. The resulting graph typically has $\Omega(i)$
connected components, while few remaining values lie within distance $c/i$
of an anchor value. Every completion must then contain $\Omega(i)$ edges
of length $\Omega(1/i)$. Integrating over suffix sizes that double at each step gives cost
$\Omega(\log n)$ except with probability $e^{-\Omega(\sqrt n)}$.
Averaging over the uniform sample then selects a fixed input with the
required lower bound in random order.

\begin{table}[t]
\centering
\caption{Algorithms for online sorting with exactly $n$ output positions.
The first three rows give bounds on cost for independent uniform inputs.
The last two give bounds on competitive ratio for arbitrary fixed inputs
arriving in random order.}
\label{tab:scope}
\begin{tabularx}{\linewidth}{
@{}>{\raggedright\arraybackslash}p{0.22\linewidth}
>{\raggedright\arraybackslash}p{0.43\linewidth}
Y@{}}
\toprule
Work & Approach & Guarantee \\
\midrule
Hu~\cite{Hu2026}
& Adaptive buffers and recursion
& $\log n\,2^{O(\log^*n)}$ in expectation \\

Kalavas et~al.~\cite{KPT2026}
& Large bins with capacity reused between stages
& $O(\log^2 n)$ with high probability \\

Hermansen~\cite{Hermansen2026}
& Hierarchical routing between buckets of constant capacity
& $O(\log n)$ in expectation \\

Platanos and Tolias~\cite{PlatanosTolias2026}
& Empirical rank partitions
& $O(\log^2 n)$ with high probability \\

This work
& Empirical partitions and hierarchical routing
& $O(\log n)$ with high probability \\
\bottomrule
\end{tabularx}
\end{table}

\paragraph*{Related models.}
Learning-augmented sorting supplies a predicted array index with each
arrival~\cite{BerceaEtAl2026}. Other work allows more output positions than
input values~\cite{Aamand2023,AzarPanigrahiVardi2026,NirjhorWein2026,Hermansen2026}.
Online metric TSP and Euclidean TSP with arrivals in random order are related
problems~\cite{Bertram2025,PlatanosTolias2026}. Here we consider placement on
the line with exactly $n$ output positions and no predictions.

\paragraph{Organization.}
Section~\ref{sec:technical-overview} outlines the proofs, and
Section~\ref{alg:section} specifies the algorithm and its capacity schedule.
Section~\ref{sec:upper-cut} proves the upper bound, including the transfer
to fixed multisets in Section~\ref{sec:transfer} and the extension to
unknown i.i.d.\ inputs in Section~\ref{tr:iid}.
Section~\ref{lb:section} proves the lower bounds, and
Section~\ref{sec:instance-sensitive} gives the bounds for structured inputs.

\section{Technical overview}\label{sec:technical-overview}

\subsection{Algorithmic framework}

We use the strict keys defined in Section~\ref{sec:model}; the output stores
only the original values. Each interval of keys is assigned a bucket of empty
array cells, which need not be consecutive. Allocated cells form a prefix of
the physical array. Empty cells in this prefix are available for placement;
the unused suffix is reserved for later phases.

\paragraph*{Learning and growing.}
After placing a fixed number $T$ of initial arrivals, a growing phase starts
with $t=T2^j$ observed keys, where $j\ge0$. We sort pointers to these keys and
choose boundaries every $T$ ranks, together with the outer boundaries
$-\infty$ and $+\infty$. The $t/T$ resulting intervals contain nearly equal
numbers of observed keys. Intervals include their left endpoint and exclude
their right endpoint, with the virtual endpoints assigned to the outer
intervals. Pairing adjacent intervals at each level gives a perfect binary
tree, with all leaves at the same depth. Each leaf receives a bucket of
capacity $2T$, supplied by old empty cells and a fresh block from the unused
array suffix. The phase processes the next $t$ arrivals, so the nominal load
per bucket is $T$, half its capacity. The observed prefix then doubles before
the next partition is built. Section~\ref{alg:section} fixes the absolute
constant $T$ and handles small inputs by consecutive placement.

\paragraph*{Placing an arrival.}
The leaf whose interval contains an arriving key is its \emph{canonical leaf}.
Among buckets with an available cell, the algorithm chooses one at minimum
tree distance from this leaf, measured by the number of tree edges. Buckets
are indexed from left to right in key order. Distance ties are broken by
increasing bucket index, and the first empty physical cell in the chosen
bucket is filled. We use Hermansen's routing rule~\cite{Hermansen2026} with
empirical partitions learned from prefixes, as in Platanos and
Tolias~\cite{PlatanosTolias2026}.

\paragraph*{Reassigning empty cells.}
Each available empty cell carries its bucket's left endpoint as a temporary
label. At a partition change, we prescribe the new capacities and process old
empty cells in physical order. The stored label determines a cell's canonical
leaf in the new tree. The same routing rule assigns it to a bucket with unused
\emph{assignment capacity}, meaning the number of additional empty cells that
the bucket may receive. We replace the label by the new bucket's left endpoint.
Fresh cells are then appended in physical order and assigned to the remaining
capacities in increasing bucket order. Stored labels define the requests even
when consecutive partitions are not nested. Only the ownership and labels of
empty cells change; filled values remain in place.

\paragraph*{Freezing and coarsening.}
Once a constant fraction of the input has arrived, we build one final empirical
tree. Each later shrinking phase merges adjacent pairs of intervals. Old empty
cells and fresh cells provide buckets of constant capacity at the coarser
level. A nonfinal phase places three quarters of its available capacity, up to
rounding, and carries the remaining cells forward. The final phase has one
bucket and only a constant number of cells left to fill. There are
$O(\log n)$ phases in each regime. Their lengths and allocation sizes depend
only on $n$, not on the placements. The integer schedule and its feasibility
are proved in Section~\ref{alg:section}.

\subsection{Upper-bound analysis}

We first take the inputs to be independent uniform samples from $[0,1]$.
We view the completed array $A$ as a path, with an edge
$(A[i],A[i+1])$ of length $|A[i+1]-A[i]|$ for each $1\le i<n$.
Fix a deterministic threshold $q\in(0,1)$, and let $C(q)$ count adjacent
output pairs on opposite sides of $q$. Almost surely, the input keys are
distinct and none equals $q$. We prove
\begin{equation}\label{overview:fixed-cut-target}
 \E[e^{\theta C(q)}]\le A_0n^{\gamma_0},
\end{equation}
where $\theta>0$, $A_0\ge1$, and $\gamma_0\ge0$ are absolute constants.
Exponential Markov inequality then gives an exponential tail for $C(q)$
above a constant multiple of $\log n$.

{We separate the main charges into growing-phase contributions,
persistent partition errors, and future fluctuations.
Growing contributions are combined by reverse conditioning.
During shrinking, charges from partition errors are summed
geometrically across merges, while future fluctuations are controlled
using independent input blocks conditional on the final growing prefix.}

\paragraph*{A capacity witness at one cut.}
For the accounting, replace each key or temporary label $x$ by the bit
$\mathbf{1}[x>q]$. Replacing a label $y$ by an arrival or a new label $z$
affects at most two adjacent edges. We charge this replacement
$2|\mathbf{1}[z>q]-\mathbf{1}[y>q]|$, which bounds its increase in crossings.
A fresh physical block starts with ordered labels and contributes at most
one initial internal crossing. We also charge at most one crossing for its
boundary with the preceding block.

An \emph{operation} consists of all arrivals in a phase or all old cell
reassignments at a transition. A request's \emph{placing node} is the least
common ancestor of its canonical and destination leaves, or that leaf if
they coincide. Mark these nodes and group each request under its highest
marked ancestor. Such a highest marked subtree exchanges no requests with
its exterior, since any such request would mark a proper ancestor. Its
requests use distinct cells or assignment slots, so their number is at most
the subtree's capacity at the start of the operation.

For an internal highest mark $v$, consider the first request crossing between
its children. The child containing its canonical leaf must be full: any
available capacity there would be closer in the tree. Before this request,
nothing entered that child from its sibling or from outside $v$. All requests
already assigned there are therefore canonical in that child. The crossing
request makes canonical demand exceed capacity. This is the \emph{local
overflow witness}.

A request grouped under a node has both its old and new key or label in that
node's interval. At each tree level, only the node containing $q$ can therefore
contribute to the charge; call it \emph{selected}. A selected node with $h$
leaves has $O(h)$ capacity and charge. Apart from the leaf contribution, a
nonzero group charge requires a local overflow witness. \begingroup
Let $\mathcal E_h$ be the event that such a witness occurs at
scale $h\in\{1,2,4,\ldots\}$, and let $h_{\max}$ be the largest scale for which
$\mathcal E_h$ occurs, or zero if none occurs. Then
\[
 \sum_{h\in\{1,2,4,\ldots\}}h\,\mathbf1[\mathcal E_h]\le2h_{\max}.
\]
Since $h_{\max}=h>0$ implies $\mathcal E_h$,
bounds of the form $\Prb(\mathcal E_h)\le C_*e^{-ch}$,
for absolute $C_*,c>0$, give a bounded exponential moment
for the weighted sum.\endgroup\ 
No independence between scales is
needed. The same argument works conditionally, and any fixed number of small
scales can be absorbed into a constant charge. The required tail estimates are proved in Sections~\ref{prob:growing} and~\ref{prob:shrinking}.

\paragraph*{Overlapping growing prefixes.}
Let $L$ be the number of growing phases, indexed by $j=0,\ldots,L-1$, and
write $P_j$ for the unordered set of keys observed at the start of phase $j$.
Thus $P_{j+1}$ consists of $P_j$ and that phase's arrivals, with
$|P_{j+1}|=2|P_j|$. The set $P_L$ defines the final empirical tree. Pair the
arrival charge of phase $j$ with the following reallocation; the last pair
includes the first shrinking reallocation.

To bound a reallocation, count all labels assigned at the start of the
preceding phase, including labels on cells that were later filled. Finite
labels occur every $T$ prefix ranks and each has multiplicity $2T$. In a new
interval $I$, their total multiplicity, including the possible label $-\infty$,
is at most $2|P_j\cap I|+4T$. This bounds surviving requests without referring
to which cells remain empty. Together with the arrival witnesses, it gives
a bound $G_j$ for the paired charge that depends only on $(P_j,P_{j+1})$.

Conditional on $P_{j+1},\ldots,P_L$, the set $P_j$ is a uniform subset of
$P_{j+1}$ of size $|P_{j+1}|/2$. For arrival witnesses, the selected intervals
are defined by this random subset. We first map the ordered elements of the fixed larger
prefix to equally spaced points in $[0,1]$, placing the threshold in the
corresponding gap. This preserves interval memberships and overflow events.
For a child with $s$ leaves in the current $g_j=|P_j|/T$-leaf tree, put $p=s/g_j$, its nominal length in these normalized rank coordinates. We localize the
selected parent to a window of length $O(p)$ around the transformed threshold,
except with exponentially small probability in the nominal sample count.
A grid with spacing a sufficiently small fixed multiple of $p$ has only
constantly many points in this window. An abnormal empirical interval forces
an atypical count in one of the fixed grid intervals. Hypergeometric tails
and a union bound give an $e^{-\Omega(h)}$ bound at scale $h$, uniformly over
the fixed larger prefix. For reallocation, the new intervals are already
fixed by $P_{j+1}$, so hypergeometric bounds apply directly to $|P_j\cap I|$.

The maximum-scale estimate gives a uniform bound on
$\E[e^{\theta_gG_j}\mid P_{j+1},\ldots,P_L]$ for an absolute $\theta_g>0$.
Successive conditioning combines these bounds, with a constant factor per
phase. The $G_j$ themselves need not be independent.

\paragraph*{Persistent partition errors.}
Suppose the tree learned from $P_L$ has $b$ finest leaves. A child containing
$s$ such leaves has nominal probability mass $s/b$; under uniform input,
its actual mass is its interval length after clipping to $[0,1]$. Call a
selected parent bad if either child's mass differs from its nominal mass
by more than a fixed small relative tolerance. We also declare scales below
a fixed absolute cutoff bad. For larger scales, the same localization and
grid argument, now using binomial counts in the uniform prefix, bounds the
error probability at a parent with ${\nu}$ finest leaves by $e^{-\Omega({\nu})}$.

While that parent remains internal in shrinking phase $i$, it contains
${\nu}/2^{i-1}$ current leaves. Its charge per operation is therefore
$O({\nu}/2^{i-1})$. Summing over phases and over arrivals and reallocations gives
$O({\nu})$, even if the error persists through many merges. Let $D_{\rm err}$
bound these total charges, excluding the first shrinking reallocation,
which is already in $G_{L-1}$. This variable depends only on $P_L$.
The error probabilities and the maximum-scale estimate give a bounded exponential
moment for $D_{\rm err}$ after averaging over $P_L$.

\paragraph*{Surviving empty cells.}
Condition on $P_L$, fixing all shrinking intervals and their good or bad
status. Consider a node interval $I$ after the merge into shrinking phase $i\ge2$ and its
corresponding subtree in the preceding phase. Write $\operatorname{cap}_{\rm old}(I)$ for
that subtree's initial capacity, $N_{i-1}(I)$ for its canonical arrivals during
the preceding phase, and $N_i^{\mathrm{re}}(I)$ for the number of its surviving empty cells. If a canonical arrival is routed
outside $I$, the subtree is already full and $N_i^{\mathrm{re}}(I)=0$. Otherwise all its
canonical arrivals occupy distinct cells inside it. Hence
\[
 N_i^{\mathrm{re}}(I)\le\max\{\operatorname{cap}_{\rm old}(I)-N_{i-1}(I),0\}.
\]
After merging, these remaining cells are precisely the old requests canonical
in $I$. If its new capacity is $\operatorname{cap}_{\rm new}(I)$, local reallocation overflow
requires
\[
 N_{i-1}(I)<\operatorname{cap}_{\rm old}(I)-\operatorname{cap}_{\rm new}(I).
\]
Thus we need a lower tail for an input count, rather than the distribution
of the surviving empty cells.

For a merge into a nonfinal phase, the unrounded schedule gives new capacity
one half of the old capacity and nominal arrival load three quarters of the
old capacity. The reallocation overflow threshold is then one half of the
old capacity, below the nominal load by a fixed fraction. The actual schedule
preserves this gap on good intervals despite rounding. Their arrival counts likewise have means below the current
capacity by a fixed fraction. Conditional on $P_L$, these counts are
binomial because the phase lengths are deterministic and the later inputs
remain independent uniforms. Chernoff bounds give exponentially small
probabilities for both kinds of witnesses.

Let $k$ be the number of shrinking phases, indexed by $i=1,\ldots,k$.
For $i<k$, let $W_i$ be the bound obtained from counts in the phase-$i$ input
block for its good-parent arrival charges and the next reallocation charges,
when that reallocation is nonfinal. The maximum-scale estimate bounds the
conditional exponential moment of each part. The two parts use the same
block and may be dependent; conditional Cauchy--Schwarz combines their
moments. Each $W_i$ depends only on $P_L$ and its own input block, so the
$W_i$ are independent conditional on $P_L$.

\paragraph*{Combining the charges.}
For an absolute constant $D_1$, the accounting gives
\begin{equation}\label{overview:decomposition}
 C(q)\le D_1(L+k+1)+\sum_{j=0}^{L-1}G_j
             +D_{\rm err}+\sum_{i=1}^{k-1}W_i.
\end{equation}
The first term covers initialization, fresh blocks, leaf charges, and the
final phase of constant size. Choose $\theta>0$ small enough for all the
moment bounds above. Conditional on $P_L$, the growing sum depends on the
ordering of that prefix, independently of the future blocks determining
the $W_i$, while $D_{\rm err}$ is fixed. The conditional moments of the two
sums are at most $M_g^L$ and $M_s^{k-1}$ for absolute constants $M_g,M_s$.
If $M_b$ bounds the moment of $D_{\rm err}$, then
\[
 \E[e^{\theta C(q)}]
 \le e^{\theta D_1(L+k+1)}M_g^LM_s^{k-1}M_b
 =n^{O(1)}.
\]
Since $L+k=O(\log n)$, this proves \eqref{overview:fixed-cut-target};
increasing $A_0$ covers the finitely many small input sizes.

\subsection{From crossings to metric cost}

Fix a strict ordering of the input keys, and let
$x_{(1)}\le\cdots\le x_{(n)}$ be their corresponding real values. Assume
$x_{(n)}>x_{(1)}$. For $1\le a<n$, let $C_a$ count adjacent output pairs with
one key rank at most $a$ and the other greater than $a$. An edge pays exactly
the sorted gaps between its endpoints. Dividing by the input range gives
\begin{equation}\label{overview:gap-identity}
 \rho(A)=\sum_{a=1}^{n-1}w_aC_a,
 \qquad
 w_a=\frac{x_{(a+1)}-x_{(a)}}{x_{(n)}-x_{(1)}},
 \qquad w_a\ge0,\quad\sum_a w_a=1.
\end{equation}

Figure~\ref{fig:crossing-identity} illustrates how edge lengths are recovered from crossing counts.\par 
\begin{figure}[t]
\centering
\begingroup
\definecolor{crossedge}{RGB}{30,85,132}
\definecolor{crosscut}{RGB}{166,62,28}
\tikzset{
  cross axis/.style={black!65,line width=.4pt},
  cross edge/.style={crossedge,line width=.85pt},
  cross cut/.style={crosscut,dashed,line width=.6pt},
  cross point/.style={circle,inner sep=0pt,minimum size=3.5pt,
    draw=crossedge,line width=.6pt},
  every node/.style={font=\small,inner sep=1pt}
}
\begin{minipage}[t]{.49\linewidth}
\centering
\begin{tikzpicture}[x=\dimexpr\linewidth/7\relax,y=1cm]
  \path[use as bounding box] (0,-.7) rectangle (7,3.05);
  \node[anchor=west] at (.08,2.8) {(a)};
  \node[anchor=east] at (1.85,2.8) {$A=$};
  \foreach \i/\v in {0/1,1/4,2/2,3/5}{
    \draw[fill=black!22,draw=black!65,line width=.4pt]
      ({2+.62*\i},2.63) rectangle ({2+.62*(\i+1)},2.97);
    \node at ({2+.62*(\i+.5)},2.8) {$\v$};
  }
  \foreach \l/\r/\y/\a/\b in {1/4/1.65/1/4,2/4/1.05/4/2,2/5/.45/2/5}{
    \node[anchor=east] at (.9,\y) {$(\a,\b)$};
    \draw[cross edge] ({\l+.4},\y)--({\r+.4},\y);
    \node[cross point,fill=crossedge] at ({\l+.4},\y) {};
    \node[cross point,fill=white] at ({\r+.4},\y) {};
  }
  \draw[cross cut] (3.4,0)--(3.4,2.03);
  \node[above,text=black] at (3.4,2.03) {$q=3$};
  \foreach \y in {1.65,1.05,.45}{
    \fill[crosscut] (3.4,\y) circle[radius=1.8pt];
  }
  \draw[cross axis,->] (1.15,0)--(5.85,0);
  \foreach \v in {1,2,4,5}{
    \draw[cross axis] ({\v+.4},-.045)--({\v+.4},.045);
    \node[below=3pt] at ({\v+.4},0) {$\v$};
  }
  \node[anchor=west] at (5.9,0) {value};
\end{tikzpicture}
\end{minipage}\hfill
\begin{minipage}[t]{.49\linewidth}
\centering
\begin{tikzpicture}[x=\dimexpr\linewidth/7\relax,y=1cm]
  \path[use as bounding box] (0,-.7) rectangle (7,3.05);
  \node[anchor=west] at (.08,2.8) {(b)};
  \node at (3.4,2.8) {$C(3)=C_2=3$};
  \fill[crossedge!8] (1.4,0)--(1.4,.55)--(2.4,.55)
    --(2.4,1.65)--(4.4,1.65)--(4.4,.55)--(5.4,.55)--(5.4,0)--cycle;
  \draw[cross axis,->] (1.15,0)--(5.85,0);
  \draw[cross axis,->] (1.15,0)--(1.15,2.08);
  \node[above] at (1.15,2.08) {$C(q)$};
  \foreach \c in {1,3}{
    \draw[cross axis] (1.105,{.55*\c})--(1.195,{.55*\c});
    \node[left=3pt] at (1.15,{.55*\c}) {$\c$};
  }
  \node[below left=2pt] at (1.15,0) {$0$};
  \draw[cross edge] (1.4,.55)--(2.4,.55);
  \draw[cross edge] (2.4,1.65)--(4.4,1.65);
  \draw[cross edge] (4.4,.55)--(5.4,.55);
  \foreach \x/\y in {1.4/.55,2.4/1.65,4.4/.55,5.4/0}{
    \node[cross point,fill=crossedge] at (\x,\y) {};
  }
  \foreach \x/\y in {2.4/.55,4.4/1.65,5.4/.55}{
    \node[cross point,fill=white] at (\x,\y) {};
  }
  \draw[cross cut] (3.4,0)--(3.4,2.03);
  \fill[crosscut] (3.4,1.65) circle[radius=1.8pt];
  \node[above,text=black] at (3.4,2.03) {$q=3$};
  \foreach \v in {1,2,4,5}{
    \draw[cross axis] ({\v+.4},-.045)--({\v+.4},.045);
    \node[below=3pt] at ({\v+.4},0) {$\v$};
  }
  \node[anchor=west] at (5.9,0) {$q$};
  \foreach \x/\a/\c in {1.9/1/1,3.4/2/3,4.9/3/1}{
    \node at (\x,-.6) {$C_\a=\c$};
  }
\end{tikzpicture}
\end{minipage}
\par\vspace{2mm}
{\small
$\displaystyle
\sum_{i=1}^{3}|A[i+1]-A[i]|=3+2+3
=\int_{1}^{5}C(q)\,dq
=1\cdot1+2\cdot3+1\cdot1=8.
$\par}
\endgroup
\caption{Crossings for $A=(1,4,2,5)$.
(a) Each edge $(A[i],A[i+1])$ is represented by the half-open interval
between its endpoint values, including the smaller endpoint.
(b) $C(q)$ counts the intervals containing $q$ and equals $C_a$ on the gap
$[x_{(a)},x_{(a+1)})$ between consecutive sorted values.
The total edge length equals the shaded area under $C(q)$,
and $\rho(A)=8/(5-1)=2$.}
\label{fig:crossing-identity}
\end{figure}
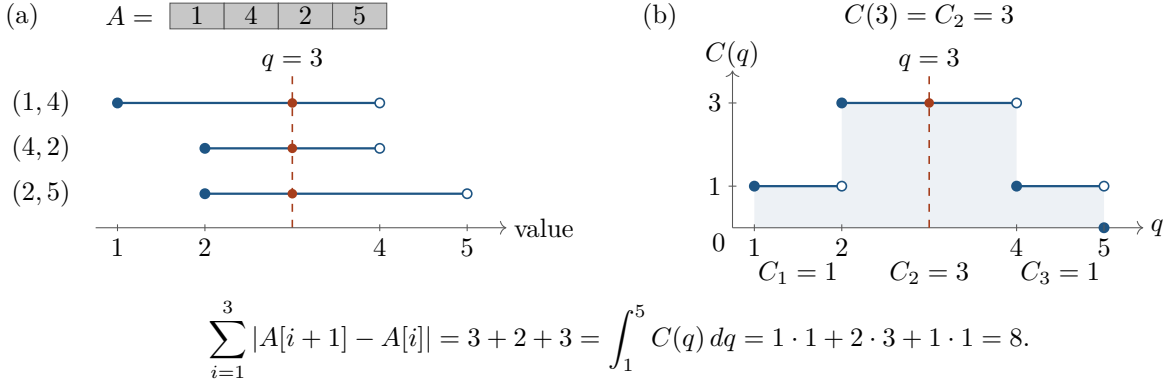
\par To bound $C_a$, return to the auxiliary uniform experiment and take $q=a/n$.
Conditional on exactly $a$ observations lying below $q$, we have $C(q)=C_a$.
The arrival permutation of the ranks is independent of the sorted values,
so it remains uniform under this conditioning. The event has probability
at least $1/(n+1)$ because $a$ is a mode of $\operatorname{Bin}(n,a/n)$.
All partition boundaries and placement decisions use only the relative order
of the keys. The law of $C_a$ is therefore the same for every fixed set of
strict keys arriving in random order. Conditioning costs a factor of at most
$n+1$, giving $\E_\pi[e^{\theta C_a}]\le2A_0n^{\gamma_0+1}$.
By convexity,
\[
 \E[e^{\theta\rho(A)}]
 \le\sum_a w_a\E[e^{\theta C_a}]
 \le2A_0n^{\gamma_0+1}.
\]
Exponential Markov inequality gives the stated upper tail for the
competitive ratio.

For repeated values, occurrence $i$ receives priority
$\xi_i=\sigma(\pi^{-1}(i))$, using the independent priority permutation
$\sigma$ from the model. Conditional on any arrival permutation $\pi$,
the vector $\xi$ is uniform, so $\xi$ and $\pi$ are independent.
Conditioning on $\xi$ gives fixed strict keys $(x_i,\xi_i)$ arriving in
uniform random order. The preceding bounds apply, with zero weights for
gaps within ties, and averaging over $\xi$ preserves them. Constant inputs
have zero cost and ratio one by convention. Section~\ref{sec:instance-sensitive}
uses the same gap weights to obtain sharper bounds for particular inputs.

\subsection{Lower-bound idea}

We adapt Hu's suffix-scale argument~\cite{Hu2026} to bound realized cost.
On independent $U[0,1]$ input, freeze an arbitrary online algorithm when $i$
arrivals remain. Its $i$ empty cells have at most $2i$ occupied neighbors;
call these neighbors \emph{anchors}. Conditional on the observed prefix
and the entire random tape, which is independent of the input, the anchors
are fixed and the suffix is still an independent uniform sample.

For a distance threshold $\delta>0$, join two suffix occurrences when their
values differ by at most $\delta$. Let $\operatorname{cc}_i(\delta)$ be the
number of components, $\operatorname{near}_i(\delta)$ the number of suffix
points within $\delta$ of an anchor value, and $\operatorname{long}_A(\delta)$
the number of final adjacent edges longer than $\delta$. Every completion
satisfies
\begin{equation}\label{overview:lower-certificate}
 \operatorname{long}_A(\delta)
 \ge\operatorname{cc}_i(\delta)-\operatorname{near}_i(\delta)-1.
\end{equation}
To prove this, delete the long edges. Within each remaining block, every
contiguous run of suffix positions lies in one proximity component. Each
run after the first in its block starts next to a frozen anchor, and its
first point is counted by $\operatorname{near}_i(\delta)$. These points are
distinct, and every proximity component appears in at least one run. There
is at most one uncharged run per block, giving the inequality.

At $\delta_i=1/(100i)$, a newly exposed suffix point lies within $\delta_i$
of an earlier point with conditional probability at most
$2\delta_i(i-1)\le1/50$. It can join at most two existing components: three
would have representatives in an interval of length $2\delta_i$, forcing
two of them to be connected already. The anchor neighborhoods have total
length at most $2\delta_i\cdot2i=1/25$. Exponential tail bounds give at least
$4i/5$ components and at most $i/10$ points near anchors, except with
probability $e^{-\Omega(i)}$. For $i\ge5$, every completion then has at
least $i/2$ long edges.

Take suffix sizes that double from order $\sqrt n$ to order $n$. The
threshold intervals $[\delta_i/2,\delta_i]$ have disjoint interiors. On the
event just described, each contributes at least
$(i/2)(\delta_i/2)=1/400$ to
$\int_0^1\operatorname{long}_A(u)\,du=\cost(A)$.
Summing over these scales gives $\Omega(\log n)$ cost. A union bound over
the exceptional events gives failure probability $e^{-\Omega(\sqrt n)}$.

For a fixed set of values, let its failure probability be taken over both
the arrival permutation and the algorithm's random tape. Averaging this
probability over independent uniform values gives the uniform-input failure
bound, because an independent shuffle preserves their joint law. Some fixed
set of distinct values therefore has at most this failure probability. It
is chosen before
either random choice, and its range is at most one, so the cost lower bound
also holds for the competitive ratio.
\section{The empirical-partition algorithm}
\label{alg:section}\label{sec:overview}

We now specify \EmpiricalSort{}, including the integer schedule and the
rules for reusing empty cells. The growing phases adapt the empirical-prefix
partitions of Platanos and Tolias~\cite{PlatanosTolias2026}; the shrinking
schedule and tree-distance placement rule follow
Hermansen~\cite{Hermansen2026}. All phase lengths and allocation sizes
are determined by $n$.

We work with strictly ordered keys and their associated real values.
For repeated values, these are the lexicographic tie-breaking keys from
Section~\ref{tr:metric-tail}; the output stores only the original values.
Partition boundaries are pointers to observed keys, together with the
ordered tags $-\infty$ and $+\infty$.

\subsection{Partitions, placements, and reallocations}
\label{alg:routing}

\paragraph*{Partition trees.}
Given a prefix of $m$ observed keys, sort pointers to them in the order
$y_1<\cdots<y_m$. For a power of two $b\le m$, choose the internal boundaries
\begin{equation}
 y_{\lfloor mj/b\rfloor},\qquad 1\le j<b,
 \label{alg:boundaries}
\end{equation}
and the outer boundaries $-\infty,+\infty$. These boundaries divide the
prefix into groups of nearly equal size without changing any output
position. Intervals are left-closed and right-open, except that the last
also includes its virtual right endpoint. In particular, the first interval
contains the label $-\infty$.

Repeatedly pairing adjacent intervals gives a perfect binary
\emph{partition tree}. The intervals used in a phase form the
\emph{current partition} and are the leaves of the \emph{current tree}.
During shrinking, this tree is obtained by truncating the fixed partition
tree at the appropriate level. The interval containing a key is its
\emph{canonical leaf}.

A \emph{bucket} is a set of currently empty physical cells assigned to one
leaf. Its cells need not be consecutive. Buckets are indexed from left to
right in key order, and the distance between two buckets is the number of
tree edges between their leaves. Thus a bucket's interval determines which
keys request it, while its cells determine where keys can be stored.

\paragraph*{Placing arrivals.}
The allocated cells always form a physical prefix of the original array.
Its empty cells are \emph{available}; cells in the unallocated suffix are
reserved for later phases. At the start of a phase, every available cell
belongs to a bucket and carries that bucket's left endpoint as a temporary
label.

For an arriving key, choose a bucket with a free cell at minimum tree
distance from its canonical leaf. Break distance ties by increasing bucket
index, and use the first free cell of the chosen bucket in physical order.
Store the associated real value irrevocably in that cell. Equivalently,
ascend from the canonical leaf to the first ancestor whose subtree contains
available capacity, then choose the least-indexed nonempty bucket in that
subtree. The routine $\textsc{Place}(t)$ processes the next $t$ arrivals by
this rule, placing each before reading the next.

\paragraph*{Reallocating empty cells.}
At a phase change, prescribe the new bucket capacities before assigning
cells. Process the old available empty cells in increasing physical order.
For each cell, use its stored label to find its canonical leaf in the new
tree. Assign the cell to a nearest bucket with unused assignment capacity,
breaking ties by bucket index, and replace its label by that bucket's left
endpoint. Assignment capacity counts how many more cells a bucket may
receive in this reallocation.

After all old cells have been assigned, append the prescribed fresh block
from the unallocated suffix of the original $n$-cell array. Assign its cells
in physical order to the remaining capacities, completing the lowest-indexed
bucket first, and give each cell its new bucket's left endpoint as label.
We call this operation \textsc{Reallocate}. It changes the labels and
ownership of empty cells but leaves all filled cells unchanged.

A stored label is a boundary key or a virtual endpoint, so it identifies a
unique interval in the new partition even when successive trees are not
nested. Under adjacent-pair merging, this interval is the parent of the
cell's old bucket. Each reassignment uses the label stored immediately
before that reassignment.

\subsection{Growing phases}
\label{alg:growing}

Fix $T=512$. For $n<16T=8192$, place arrivals consecutively and stop.
For $n\ge16T$, place the first $T$ arrivals in the first $T$ cells. Set
\begin{equation}
 t_j=T2^j,\qquad
 J=\max\{j\ge0:t_j\le n/8\},\qquad L=J+1.
 \label{alg:growing-times}
\end{equation}
There are $L$ growing phases, indexed by $j=0,\ldots,J$. Phase $j$
starts with $t_j$ observed keys. Build their empirical partition with
$g_j=t_j/T$ leaves and prescribe capacity $2T$ for every bucket.
The total available capacity is $2t_j$, and the phase places the next
$t_j$ arrivals.

Phase $0$ allocates $2T$ fresh cells. For $j\ge1$, the preceding phase
leaves $t_j/2$ old empty cells, so append $3t_j/2$ fresh cells.
In both cases, call \textsc{Reallocate} before $\textsc{Place}(t_j)$.
At the end of phase $j$, the numbers of observed keys, allocated cells,
and available empty cells are, respectively,
\begin{equation}
 2t_j,\qquad 3t_j,\qquad t_j.
 \label{alg:growing-counts}
\end{equation}
For $j=0$, these counts follow from the initial $T$ filled cells and the
$2T$ fresh cells. For $j\ge1$, phase $j-1$ left $t_j$ filled cells and
$t_j/2$ empty cells. Adding $3t_j/2$ cells and placing $t_j$ arrivals gives
the displayed counts. All allocations are integral, and the allocated
prefix never exceeds $3t_J\le3n/8$ cells.

Let $P_j$ be the unordered set of the first $t_j$ keys, including $j=L$.
After the last growing phase, write
\begin{equation}
 m=2t_J=t_L,\qquad N=n-m.
 \label{alg:frozen-size}
\end{equation}
There are now $m$ filled cells and $N$ arrivals still to come. Of the $N$
empty cells in the original array, $t_J$ are already allocated. The filled
cells need not occupy its first $m$ positions. We build one new partition
tree from $P_L$ and keep its boundaries throughout the shrinking phases.

Since $t_J\le n/8<t_{J+1}=m$, we have
\begin{equation}
 \frac n8<m\le\frac n4,\qquad
 \frac{3n}{4}\le N<\frac{7n}{8},\qquad m\le\frac N3.
 \label{alg:size-ratios}
\end{equation}

\subsection{Fixed partition tree and shrinking phases}
\label{alg:shrinking}

{The shrinking schedule halves the bucket count in each phase and
ends with one bucket. Between nonfinal phases, about one quarter
of the preceding capacity survives and fresh cells supply another
quarter. The total available capacity therefore roughly halves,
keeping capacity per bucket nearly constant.} Let $k$ be the number of phases, $b_i$ the bucket count in
phase $i$, and $\lambda$ the ideal bucket capacity before rounding. Set
\begin{equation}
 k=\lfloor\log_2N\rfloor-7,\qquad
 b_i=2^{k-i}\ (1\le i\le k),\qquad
 \lambda=\frac{4N}{3\cdot2^k},\qquad B=\lfloor\lambda\rfloor-3.
 \label{alg:shrinking-parameters}
\end{equation}
The quantity $B$ will bound every nonfinal bucket's capacity from below.
Here $N\ge6144$, so $k\ge2$. Moreover,
$b_1=2^{k-1}\le N/256$ and $N<7m$, so $b_1<m$ and the $m$ observed
keys suffice to build the fixed tree.

Phase $1$ uses its $b_1$ finest intervals. Each later phase merges adjacent
pairs, giving $b_i$ intervals. The surviving boundaries still have the form
\eqref{alg:boundaries} with prefix size $m$ and bucket count $b_i$:
retaining every $2^{i-1}$st finest boundary gives
$y_{\lfloor mj/b_i\rfloor}$.

\paragraph*{Phase lengths.}
Let $M_i$ be the number of available empty cells after the phase-$i$
allocation and before its placements. A nonfinal phase places
\begin{equation}
 n_i=\lfloor3M_i/4\rfloor
 \label{alg:phase-length}
\end{equation}
arrivals, leaving $\lceil M_i/4\rceil$ available cells. The final phase
places $n_k=M_k$ arrivals.

\paragraph*{Fresh cells and capacities.}
Let $m+S_i$ be the endpoint of the allocated physical prefix after the
phase-$i$ allocation. Choose
\begin{equation}
 S_0=0,\qquad S_k=N,\qquad
 S_i=\left\lfloor N-\frac{2N}{3\cdot2^i}\right\rfloor
 \quad(1\le i<k).
 \label{alg:fresh-boundaries}
\end{equation}
The convention $S_0=0$ is only auxiliary; before shrinking, the allocated
prefix has length $m+t_J$.

The first shrinking phase has $M_1=S_1=\lfloor2N/3\rfloor$ available
cells: the $t_J$ old empty cells and $S_1-t_J$ fresh cells. For $i\ge2$,
append $S_i-S_{i-1}$ fresh cells to those left by phase $i-1$, giving
\begin{equation}
 M_i=(S_i-S_{i-1})+\lceil M_{i-1}/4\rceil
 \qquad(2\le i\le k).
 \label{alg:capacity-recurrence}
\end{equation}
Thus $S_i$ specifies the cumulative allocation endpoint, while $M_i$
counts the capacity available for the current phase.

Distribute $M_i$ as evenly as possible among the $b_i$ buckets. Writing
$M_i=b_i d_i+r_i$ with $0\le r_i<b_i$, give the first $r_i$ buckets
capacity $d_i+1$ and the rest capacity $d_i$. Reallocate the old empty
cells first, assign the fresh cells to the remaining capacities, and call
$\textsc{Place}(n_i)$. Algorithm~\ref{alg:pseudocode} summarizes the full
procedure.

\begin{algorithm}[tbp]
\caption{\EmpiricalSort}
\label{alg:pseudocode}\label{alg:main-summary}
\begin{algorithmic}[1]
\Require An online stream of $n$ keys and an array of $n$ cells
\State $T\gets512$
\If{$n<16T$}
  \State Place all arrivals consecutively; \Return
\EndIf
\State Place the first $T$ arrivals in the first $T$ cells
\State Compute $t_j,J,L$ from~\eqref{alg:growing-times}
\For{$j=0,\ldots,J$} \hfill \Comment{Growing phases}
  \State Build the $t_j/T$-leaf empirical tree from the first $t_j$ keys
  \State $f\gets2T$ if $j=0$, and $f\gets3t_j/2$ otherwise
  \State \Call{Reallocate}{current tree, capacity $2T$ per bucket, $f$ fresh cells}
  \State \Call{Place}{$t_j$}
\EndFor
\State Set $m=2t_J$, $N=n-m$, and compute $k,b_i,B,S_i$ from
\Statex \hspace{\algorithmicindent}\eqref{alg:shrinking-parameters}
        and~\eqref{alg:fresh-boundaries}
\State Build the fixed $b_1$-leaf empirical tree from the first $m$ keys
\For{$i=1,\ldots,k$} \hfill \Comment{Shrinking phases}
  \State Use the finest leaves if $i=1$; otherwise merge adjacent pairs
  \State $f\gets S_1-t_J$ if $i=1$, and $f\gets S_i-S_{i-1}$ otherwise
  \State $M_i\gets$ number of old available empty cells $+f$
  \State \Call{Reallocate}{current tree, balanced capacities totaling $M_i$, $f$}
  \State $n_i\gets\lfloor3M_i/4\rfloor$ if $i<k$, and $n_i\gets M_i$ otherwise
  \State \Call{Place}{$n_i$}
\EndFor
\end{algorithmic}
\end{algorithm}

Figure~\ref{fig:empiricalsort-execution} illustrates the growing phases and the transition to the shrinking phases.\par 
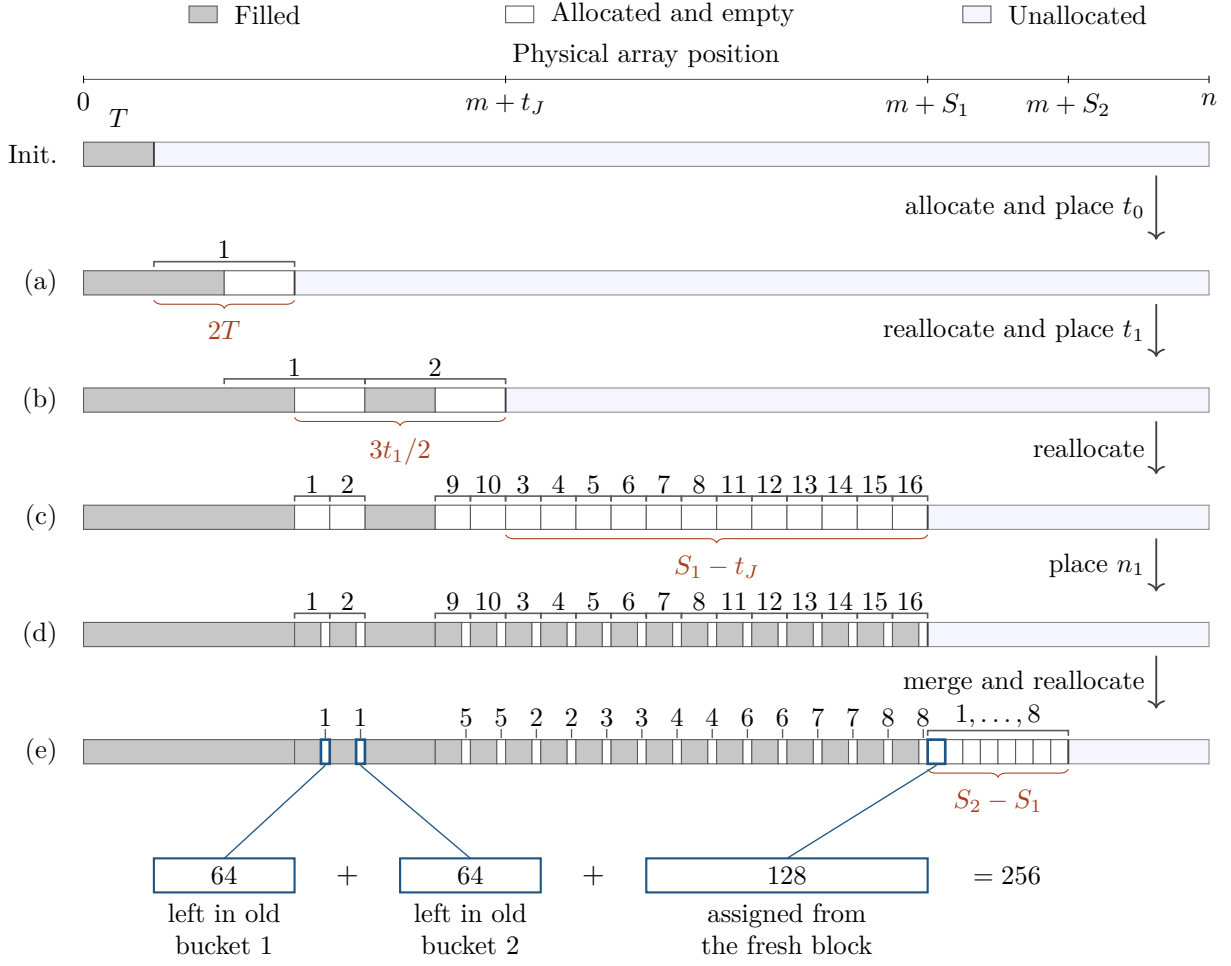
\begin{figure}[tbp]
\centering
\begingroup
\definecolor{exbucket}{RGB}{30,85,132}
\definecolor{exfresh}{RGB}{166,62,28}
\newcommand{\exrun}[4]{\draw[fill=#4,draw=black!65,line width=.4pt] (#1,#3) rectangle (#2,{#3+.32});}
\newcommand{\exbucket}[4]{\draw[black!65,line width=.6pt] (#1,{#3+.37})--(#1,{#3+.44})--(#2,{#3+.44})--(#2,{#3+.37});\node[font=\small,text=black,above,inner sep=1pt] at ({(#1+#2)/2},{#3+.44}) {#4};}
\newcommand{\exrest}[2]{\draw[fill=blue!4,draw=black!45,line width=.4pt] (#1,#2) rectangle (32,{#2+.32});\draw[black!75,line width=.7pt] (#1,#2)--(#1,{#2+.32});}
\newcommand{\exfresh}[4]{\draw[exfresh,decorate,decoration={brace,mirror,amplitude=3pt}] (#1,{#3-.07})--(#2,{#3-.07}) node[midway,below=4pt,font=\small]{#4};}
\begin{tikzpicture}[x=.465cm,y=1cm]
\foreach \x/\shade/\txt in {3/{black!22}/Filled,12/white/{Allocated and empty},25/{blue!4}/Unallocated}{
\draw[fill=\shade,draw=black!65] (\x,3.4) rectangle ({\x+.8},3.62);
\node[anchor=west,font=\small] at ({\x+1},3.51){\txt};}
\draw[black!65] (0,2.65)--(32,2.65);
\foreach \x/\t in {0/0,12/{m+t_J},24/{m+S_1},28/{m+S_2},32/n}{\draw (\x,2.6)--(\x,2.7);\node[font=\small,below] at (\x,2.6){$\t$};}
\node[font=\small,above] at (16,2.68) {Physical array position};
\node[anchor=east,font=\small] at (-.45,1.66) {Init.};
\exrun{0}{2}{1.5}{black!22}\exrest{2}{1.5}
\node[font=\small,above] at (1,1.88) {$T$};
\foreach \ya/\yb/\txt in {1.38/.53/{allocate and place $t_0$},-.32/-1.02/{reallocate and place $t_1$},-1.87/-2.57/{reallocate},-3.42/-4.12/{place $n_1$},-4.97/-5.67/{merge and reallocate}}{
\draw[->,black!75,line width=.7pt] (30.5,\ya)--(30.5,\yb) node[midway,left=3pt,font=\small,text=black,fill=white,inner sep=1pt]{\txt};}
\foreach \y/\p in {-.2/a,-1.75/b,-3.3/c,-4.85/d,-6.4/e}{\node[anchor=east,font=\small] at (-.45,{\y+.16}) {(\p)};}
\begin{scope}[yshift=-2mm] \exrun{0}{4}{0}{black!22}\exrun{4}{6}{0}{white}\exrest{6}{0}
\exbucket{2}{6}{0}{1}\exfresh{2}{6}{0}{$2T$}
\begin{scope}[yshift=-2mm] \exrun{0}{6}{-1.35}{black!22}\exrun{6}{8}{-1.35}{white}
\exrun{8}{10}{-1.35}{black!22}\exrun{10}{12}{-1.35}{white}\exrest{12}{-1.35}
\exbucket{4}{8}{-1.35}{1}\exbucket{8}{12}{-1.35}{2}\exfresh{6}{12}{-1.35}{$3t_1/2$}
\begin{scope}[yshift=-2mm] \exrun{0}{6}{-2.7}{black!22}\exrun{8}{10}{-2.7}{black!22}
\foreach \x/\b in {6/1,7/2,10/9,11/10,12/3,13/4,14/5,15/6,16/7,17/8,18/11,19/12,20/13,21/14,22/15,23/16}{\exrun{\x}{\x+1}{-2.7}{white}\exbucket{\x}{\x+1}{-2.7}{\b}}
\exrest{24}{-2.7}\exfresh{12}{24}{-2.7}{$S_1-t_J$}
\begin{scope}[yshift=-2mm] \exrun{0}{6}{-4.05}{black!22}\exrun{8}{10}{-4.05}{black!22}
\foreach \x/\b in {6/1,7/2,10/9,11/10,12/3,13/4,14/5,15/6,16/7,17/8,18/11,19/12,20/13,21/14,22/15,23/16}{\exrun{\x}{\x+.75}{-4.05}{black!22}\exrun{\x+.75}{\x+1}{-4.05}{white}\exbucket{\x}{\x+1}{-4.05}{\b}}
\exrest{24}{-4.05}
\begin{scope}[yshift=-2mm] \exrun{0}{6}{-5.4}{black!22}\exrun{8}{10}{-5.4}{black!22}
\foreach \x in {6,7,10,11,12,13,14,15,16,17,18,19,20,21,22,23}{\exrun{\x}{\x+.75}{-5.4}{black!22}\exrun{\x+.75}{\x+1}{-5.4}{white}}
\foreach \x/\p in {6/1,10/5,12/2,14/3,16/4,18/6,20/7,22/8}{
\foreach \d in {0,1}{
\draw[black!65,line width=.5pt] ({\x+\d+.875},-5.06)--({\x+\d+.875},-4.94);
\node[font=\small,text=black,above,inner sep=1pt] at ({\x+\d+.875},-4.94){\p};}}
\foreach \p in {1,...,8}{\pgfmathsetmacro{\xx}{24+(\p-1)/2}\exrun{\xx}{\xx+.5}{-5.4}{white}}
\exbucket{24}{28}{-5.4}{$1,\ldots,8$}
\exrest{28}{-5.4}\exfresh{24}{28}{-5.4}{$S_2-S_1$}
\foreach \l/\r in {6.75/7,7.75/8,24/24.5}{\draw[exbucket,line width=1pt] (\l,-5.4) rectangle (\r,-5.08);}
\foreach \a/\b in {6.875/4,7.875/11,24.25/20}{
\draw[exbucket,line width=.5pt] (\a,-5.42)--(\b,-6.65);}
\draw[draw=exbucket,fill=white,line width=.8pt] (2,-7.1) rectangle (6,-6.65);
\draw[draw=exbucket,fill=white,line width=.8pt] (9,-7.1) rectangle (13,-6.65);
\draw[draw=exbucket,fill=white,line width=.8pt] (16,-7.1) rectangle (24,-6.65);
\node[font=\small] at (4,-6.875) {64};\node[font=\small] at (11,-6.875) {64};\node[font=\small] at (20,-6.875) {128};
\node at (7.5,-6.875) {$+$};\node at (14.5,-6.875) {$+$};\node[font=\small,anchor=west] at (25,-6.875) {$=256$};
\node[font=\small,below,align=center] at (4,-7.12) {left in old\\bucket 1};\node[font=\small,below,align=center] at (11,-7.12) {left in old\\bucket 2};\node[font=\small,below,align=center] at (20,-7.12) {assigned from\\the fresh block};
\end{scope}\end{scope}\end{scope}\end{scope}\end{scope} \end{tikzpicture}
\endgroup
\caption{An example of the growing and shrinking phases with $n=8192$ and
$T=512$ ($J=1$, $L=2$, $k=5$).
Numbers above the arrays identify buckets. Red braces mark fresh blocks.
Initially, the first $T$ arrivals are placed in consecutive cells.
(a) Growing phase 0 allocates fresh cells for one bucket and places $t_0$ arrivals.
(b) Growing phase 1 reassigns the remaining empty cells to two buckets,
allocates and assigns fresh cells, and places $t_1$ arrivals.
(c) The observed values define 16 new bucket intervals. In this example,
the empty cells remaining from old bucket 1 are assigned to new buckets 1 and 2,
and those from old bucket 2 to new buckets 9 and 10.
Fresh cells are assigned to meet the remaining bucket capacities.
(d) After placing the arrivals, each bucket is three quarters full in this example.
(e) The next shrinking phase merges adjacent bucket intervals,
reassigns the remaining empty cells, and allocates and assigns fresh cells.
Each number in (e) labels only the empty piece directly below it.
Repeated numbers indicate pieces of the same bucket, as illustrated by the
enlargement of bucket 1. All existing cells stay in place.}
\label{fig:empiricalsort-execution}
\end{figure}
  
\subsection{Capacity invariants and feasibility}
\label{alg:capacity-proof}
The following bounds will supply the capacity thresholds in the probability
analysis. They also show that the algorithm completes every input order,
regardless of which buckets overflow.

\begin{lemma}[Capacities and feasibility]
\label{alg:feasibility}
The shrinking schedule satisfies $167\le B\le338$. Every bucket in a
nonfinal shrinking phase has capacity in $[B,B+7]$, and the final phase
has $M_k<515$ available cells. The complete algorithm is legal on every
input order, uses exactly $n$ physical cells, and places all $n$ arrivals.
All phase lengths and allocation times are deterministic functions of $n$.
\end{lemma}

\begin{proof}
We first bound rounding errors and bucket capacities, then verify the
allocation and arrival counts.

\paragraph*{Rounding errors.}
For $1\le i<k$, define the ideal capacity and its error by
\[
 \overline M_i=\frac{4N}{3\cdot2^i}=\lambda b_i,
 \qquad e_i=M_i-\overline M_i.
\]
The ideal total halves at each step: one quarter survives from the preceding
phase and another quarter is supplied by fresh cells. We prove $|e_i|<3$.

Initially, $M_1=\lfloor2N/3\rfloor$ gives $-1<e_1\le0$. For $i+1<k$,
the unfloored allocation endpoints differ by
\[
 \left(N-\frac{2N}{3\cdot2^{i+1}}\right)
 -\left(N-\frac{2N}{3\cdot2^i}\right)
 =\frac{N}{3\cdot2^i}=\frac{\overline M_i}{4}.
\]
Hence
\[
 S_{i+1}-S_i=\frac{\overline M_i}{4}+\delta_i,
 \qquad |\delta_i|<1.
\]
Since $M_i$ is an integer, we can also write
\[
 \left\lceil\frac{M_i}{4}\right\rceil
 =\frac{M_i}{4}+\varepsilon_i,
 \qquad 0\le\varepsilon_i\le\frac34.
\]
The next phase combines fresh cells with surviving empty cells:
\[
\begin{aligned}
 M_{i+1}
 &=(S_{i+1}-S_i)+\left\lceil\frac{M_i}{4}\right\rceil\\
 &=\frac{\overline M_i}{4}+\delta_i
   +\frac{M_i}{4}+\varepsilon_i.
\end{aligned}
\]
Subtracting $\overline M_{i+1}=\overline M_i/2$ gives
\[
\begin{aligned}
 e_{i+1}
 &=M_{i+1}-\overline M_{i+1}\\
 &=\frac{M_i-\overline M_i}{4}+\delta_i+\varepsilon_i
 =\frac{e_i}{4}+\delta_i+\varepsilon_i.
\end{aligned}
\]
If $|e_i|<3$, then $-7/4<e_{i+1}<5/2$, completing the induction.

\paragraph*{Nonfinal bucket capacities.}
For $i<k$, we have $b_i\ge2$ and therefore
\[
 \left|\frac{M_i}{b_i}-\lambda\right|
 =\frac{|e_i|}{b_i}<\frac3{b_i}\le\frac32.
\]
Balanced allocation gives each bucket either
$\lfloor M_i/b_i\rfloor$ or $\lceil M_i/b_i\rceil$ cells.
Rounding changes $M_i/b_i$ by less than one, so the preceding
bound places every bucket capacity strictly between
\[
 \lambda-\frac32-1=\lambda-\frac52
 \qquad\text{and}\qquad
 \lambda+\frac32+1=\lambda+\frac52.
\]
Since $B=\lfloor\lambda\rfloor-3$,
\[
 \begin{aligned}
 B&=\lfloor\lambda\rfloor-3\le\lambda-3<\lambda-\frac52,\\
 \lambda+\frac52
 &<\lfloor\lambda\rfloor+\frac72
 =B+\frac{13}{2}<B+7.
\end{aligned}
\]
This proves the capacity envelope. The definition of $k$ gives
\[
 2^{k+7}\le N<2^{k+8},\qquad
 \frac{512}{3}\le\lambda <\frac{1024}{3},
\]
so $170\le\lfloor\lambda\rfloor\le341$ and $167\le B\le338$.

\paragraph*{The final phase.}
Since $S_k=N$, $2N/(3\cdot2^{k-1})=\lambda$, and $N$ is an integer,
the final allocation adds
\[
 S_k-S_{k-1}=N-\lfloor N-\lambda\rfloor=\lceil\lambda\rceil.
\]
The preceding phase has $b_{k-1}=2$ buckets, so
$\overline M_{k-1}=\lambda b_{k-1}=2\lambda$.
Together with $|e_{k-1}|<3$, this gives $M_{k-1}<2\lambda+3$.
Adding these fresh cells to the $\lceil M_{k-1}/4\rceil$ empty cells
left by the preceding phase gives
\[
 \begin{aligned}
 M_k
 &=\lceil\lambda\rceil+\lceil M_{k-1}/4\rceil\\
 &<(\lambda+1)+\frac{2\lambda+3}{4}+\frac34\\
 &=\frac32\lambda+\frac52<512+\frac52<515.
 \end{aligned}
\]

\paragraph*{Allocation and arrival counts.}
The growing counts are given by~\eqref{alg:growing-counts}. The shrinking
endpoints $S_i$ are nondecreasing and end at $N$. The first fresh block
has nonnegative size because
\[
 S_1\ge\frac{2N}{3}-1\ge\frac n2-1>t_J,
\]
using $N\ge3n/4$, $t_J\le n/8$, and $n\ge8192$. Every later fresh block
has size $S_i-S_{i-1}\ge0$.

Let $Q_{i-1}=\sum_{r<i}n_r$, with $Q_0=0$. Immediately after the
phase-$i$ allocation, we claim
\[
 \text{allocated cells}=m+S_i,\qquad
 \text{filled cells}=m+Q_{i-1},\qquad
 M_i=S_i-Q_{i-1}.
\]
For $i=1$, adding $S_1-t_J$ cells to the previous allocation of $m+t_J$
gives $m+S_1$ cells, of which $m$ are filled. After a nonfinal phase $i$,
put $Q_i=Q_{i-1}+n_i$. Its remaining capacity is
\[
 M_i-n_i=S_i-Q_i=\lceil M_i/4\rceil.
\]
Appending $S_{i+1}-S_i$ cells gives $m+S_{i+1}$ allocated cells and
$M_{i+1}=S_{i+1}-Q_i$ empty cells, proving the claim inductively.

There are $N-Q_{i-1}$ arrivals left at the start of phase $i$. For $i<k$,
\[
 n_i\le M_i=S_i-Q_{i-1}\le N-Q_{i-1},
\]
so the prescribed placements have both enough space and enough remaining
inputs. In the final phase, $S_k=N$ gives
$M_k=N-Q_{k-1}$, exactly the number of remaining arrivals. These fill the
last bucket. The allocated prefix has length at most $m+S_i\le n$ and
ends at $m+N=n$.

At each reallocation, target capacities total the number of old and fresh
empty cells. A target slot therefore remains available until all old cells
have been assigned, after which fresh cells fill the residual slots.
Likewise, every prescribed arrival has an available destination. Thus all
operations are legal. The recurrences depend only on $n$, so the same is
true of phase lengths and allocation times. The sequential branch for
$n<16T$ is immediate.
\end{proof}

\begin{remark}[Bounds used in the analysis]
\label{alg:uniform-bounds}
Every nonfinal bucket has capacity at most $1024$: growing buckets have
capacity $2T$, and shrinking buckets have capacity at most $B+7\le345$.
For a nonfinal shrinking phase,
\begin{equation}
 \frac{3B}{4}-1\le\frac{n_i}{b_i}
 \le\frac34(B+7).
 \label{alg:arrival-density}
\end{equation}
Indeed, $Bb_i\le M_i\le(B+7)b_i$ and
$3M_i/4-1\le n_i\le3M_i/4$. Dividing by $b_i$ gives the bounds.
The ratio $n_i/b_i$ is the total phase length divided by the bucket count;
individual buckets may receive different numbers of canonical arrivals.

Writing $b=b_1$, we also have
\begin{equation}
 b\le N/256,\qquad \frac{m}{b}>\frac{256}{7}>32.
 \label{alg:pilot-density}
\end{equation}
These follow from $b=2^{k-1}\le N/256$ and $m/N>1/7$ in
\eqref{alg:size-ratios}. Finally, $L+k=O(\log n)$, and at most one fresh
block is appended in each phase.
\end{remark}

\subsection{Implementation}
\label{alg:implementation}

We count key comparisons and unit-cost operations on $O(\log n)$-bit
integers. A word may also store a cell index or a key pointer. The bound
excludes generating the random tie priorities and representing arbitrary
real inputs; it includes comparisons of the resulting strict keys and all
allocation data structures.

\begin{proposition}
\label{alg:implementation-bound}
Given its tie priorities, \EmpiricalSort\ can be implemented using
$O(n\log n)$ comparisons and word operations and $O(n)$ auxiliary words.
\end{proposition}

\begin{proof}
First bound the total available capacity over all phases. Growing phases
contribute $\sum_{j=0}^J2t_j\le4t_J=O(n)$. For shrinking phases,
Lemma~\ref{alg:feasibility} gives
\begin{equation}
 \sum_{i=1}^k M_i
 <\lambda\sum_{i=1}^{k-1}b_i+3(k-1)+515
 \le\frac{4N}{3}+3k+515=O(n).
 \label{alg:total-capacity}
\end{equation}
The old empty cells at any reallocation form a subset of that phase's
available cells. Thus there are $O(n)$ old-cell reassignments in total,
although a particular cell may be reassigned more than once.

\paragraph*{Partition trees.}
Keep pointers to all observed keys. Sorting the prefix pointers in growing
phase $j$ takes $O(t_j\log(t_j+1))$ comparisons, for a total of
$O(n\log n)$ because $\sum_jt_j\le2t_J$. Sort $P_L$ once more and compute
the boundary ranks in~\eqref{alg:boundaries} by integer arithmetic.
A tree with $r$ leaves uses $O(r)$ nodes and boundary pointers.
The growing leaf counts form a geometric progression, and the shrinking
leaf counts sum to less than $2b_1$, so all tree construction takes $O(n)$
additional work. Merged trees reuse the sorted order and point directly to
the appropriate finest boundaries.

\paragraph*{Routing.}
At each node, store whether any descendant leaf has available capacity.
For arrivals this means a bucket with an empty cell; during reallocation it
means positive residual assignment capacity. Binary search among the ordered
boundaries finds a canonical leaf. If that leaf has no capacity, ascend to
the first ancestor with capacity and descend through available children,
choosing the left child whenever possible.

This returns the least-indexed nearest available bucket. The child of the
first available ancestor that contains the canonical leaf has no capacity,
so every available leaf under that ancestor has the same least common
ancestor with the canonical leaf. Leaves outside it are farther away.
A search or availability update takes $O(\log(r+1))$ operations.
Use separate residual counters during reallocation, then initialize the
arrival-availability flags from the completed buckets.

\paragraph*{Physical cell order.}
Maintain a doubly linked list of available empty cells in increasing physical
index. Remove a cell when it is filled, and append fresh cells in index order.
Reallocation scans this list without changing its order. Build each new
bucket as a queue: old cells enter in physical order, followed by fresh cells,
whose indices are larger. The queue front is therefore the first free cell
required by the placement rule.

Each cell record has separate links for the global free list and its bucket
queue. A placement removes the bucket's queue front and its global-list record
in constant time, in addition to updating the tree. All scans and queue
construction take $O(\sum_i M_i+\sum_j2t_j)=O(n)$ time.

There are $n$ placements and $O(n)$ reassignments and fresh-cell assignments,
so routing takes $O(n\log n)$ time in total. Cell records, key pointers,
sorting workspace, and the current trees occupy $O(n)$ words. Old trees can
be discarded after reallocation because surviving labels point to immutable
observed keys. The sequential small-$n$ branch takes $O(n)$ time.
\end{proof}
\section{The upper bound}
\label{sec:upper-bound}\label{sec:upper-cut}

We prove Theorem~\ref{thm:upper} for the algorithm in
Section~\ref{alg:section} by counting crossings. For the completed array
$A$ and a numerical threshold $q$, let $C(q)$ count adjacent pairs of
entries on opposite sides of $q$. An edge of length $|x-y|$ crosses precisely the
thresholds between its endpoints, so
\[
 \cost(A)=\int_{\mathbb R}C(q)\,dq.
\]
This is the one-dimensional coarea identity for a piecewise constant function
\cite{EvansGariepy2015}.
The algorithm does not use $q$.

We bound $C(q)$ using local capacity imbalances, called \emph{witnesses},
and combine their exponential moments under uniform input.
Section~\ref{sec:transfer} transfers the bound to random-order multisets
through comparison invariance and numerical gaps. Section~\ref{tr:iid}
extends it to unknown i.i.d.\ input (Corollary~\ref{cor:unknown-iid}).

\subsection{From crossing charges to local witnesses}
\label{acc:section}

We use Hermansen's temporary-value accounting and placing-node
method~\cite{Hermansen2026}, extending them to reallocations between
empirical trees that need not be nested. The bounds below are deterministic.
Recall $T=512$, with $L$ growing and $k$ shrinking phases.

\subsubsection{Replacement charges and physical blocks}
\label{acc:binary-cut}

The accounting also covers cuts in the strict key order, including
cuts between equal values distinguished by priorities. Only numerical
thresholds enter the integral. Fix a cut $q$ distinct from every key,
and encode its two sides by
\[
 \beta_q(x)=\mathbf1[x>q],\qquad
 \beta_q(-\infty)=0,\qquad \beta_q(+\infty)=1.
\]
Each allocated empty cell carries a temporary label. An arrival replaces
the label by the input key. Reallocation replaces it by a new label
without moving the cell.

An \emph{operation} consists of a phase's arrivals or its old-cell
reassignments. Each request $o$ is one arrival or old empty cell,
respectively. Let $y(o),z(o)$ be the cell's key or label before and after
replacement. Define
\[
 \chi_q(o)=2|\beta_q(z(o))-\beta_q(y(o))|,
 \qquad D_{\rm op}(q)=\sum_{o\in\mathcal O_{\rm op}}\chi_q(o),
\]
where $\mathcal O_{\rm op}$ is the request set. The factor of two accounts
for the cell's two possible neighbors. Initial labels on fresh cells are
counted separately.

\paragraph*{Physical blocks.}
\phantomsection\label{acc:physical-blocks}
The initial allocation and $F$ subsequent nonempty fresh allocations
partition the array into permanent contiguous blocks. Fresh cells are
assigned in bucket order, so each new block has nondecreasing initial labels.

\begin{lemma}[Accounting for crossings]
\label{acc:global-accounting}
For $n\ge16T$, on every execution,
\begin{equation}
 C(q)\le (T-1)+2F+\sum_{\rm op}D_{\rm op}(q),
 \qquad F\le L+k,
 \label{acc:global-bound}
\end{equation}
where the sum runs over all arrival and reallocation operations after
the initial block. The arrival and reallocation charges of the final
one-bucket phase together are less than $2060$.
\end{lemma}

\begin{proof}
If a cell changes from $y$ to $z$ and a neighbor in the same physical
block carries $a$, the triangle inequality gives
\[
 |\beta_q(a)-\beta_q(z)|-|\beta_q(a)-\beta_q(y)|
 \le |\beta_q(z)-\beta_q(y)|.
\]
Applying this to both internal neighbors and summing chronologically
bounds each block's final internal crossings by its initial crossings
plus its charges.

The initial block contributes at most $T-1$ crossings. Each fresh block starts
with zeros followed by ones, giving at most one internal crossing.
The $F$ block-boundary edges contribute at most $F$ more, yielding
$2F$. Each phase adds at most one block, so $F\le L+k$.
By Lemma~\ref{alg:feasibility}, the final phase has fewer than $515$
arrivals and fewer than $515$ reassignments, each charged at most two, totaling
less than $2060$.
\end{proof}

\subsubsection{Top-marked subtrees}
\label{sec:local-witnesses}

Fix an operation and its $b$-leaf tree. Capacities count initially
available cells for arrivals, or initially unused assignment slots for
reallocation. Write $\capc(v)$ for the total capacity below node $v$, or
$\capc(I)$ when referring to its interval $I$.

A request's \emph{canonical leaf} $a(o)$ contains its arriving key for
an arrival and its previous label for a reallocation. Let $N_{\rm op}(v)$
count requests whose canonical leaves lie below $v$. The subtree
\emph{locally overflows} when
\[
 N_{\rm op}(v)>\capc(v).
\]
Requests from outside can fill a subtree without causing such an excess.

Mark each request's \emph{placing node}: the least common ancestor of
its canonical and destination leaves, or that leaf when they coincide.
After the operation, a marked node is \emph{top-marked} if it has no
marked proper ancestor. The subtrees rooted at top-marked nodes are
pairwise disjoint. Assign each request to
its placing node's highest marked ancestor, possibly the node itself. Let
$\Gamma_q(v)$ sum the charges assigned to $v$ and $\mathcal M_{\rm op}$
denote the top marks. Then
\[
 D_{\rm op}(q)=\sum_{v\in\mathcal M_{\rm op}}\Gamma_q(v).
\]
\begin{lemma}[Top-marked subtrees]\label{acc:topmarks}
For one operation type in one phase, at most $\capc(v)$ requests are
charged to a top-marked node $v$. If $v$ is not a leaf, one of its children
locally overflows. These conclusions apply to reallocation even when the
old and new partition trees are not nested.
\end{lemma}

\phantomsection\label{proof:topmarks}
\begin{proof}
A request entering or leaving $v$ would mark a proper ancestor,
contradicting that $v$ is top-marked. Assigned requests use distinct slots
inside $v$, giving the capacity bound.

For internal $v$, let $u$ be the canonical child of the first request
crossing between $v$'s children. The routing rule requires $u$ to be full,
since any nonfull bucket inside $u$ would be closer to the canonical leaf.
Before this crossing, no request has entered $u$ from its sibling, and none can have entered
from outside $v$. Thus its $\capc(u)$ occupants are canonical in $u$.
The crossing request is one more such request, so
$N_{\rm op}(u)>\capc(u)$.

For reallocation, use the new tree's initially unused assignment slots.
Fresh cells are assigned only after all old-cell requests.
\end{proof}

\paragraph*{Only the intervals containing the cut contribute.}
For a node $v$ with interval endpoints $\ell(v),r(v)$, define
\[
 \operatorname{span}_q(v)
 =\beta_q(r(v))-\beta_q(\ell(v))\in\{0,1\}.
\]
This indicates crossing, not numerical length. Exactly one node per
level contains $q$ and has span one. These \emph{selected nodes} form a
chain, including outer cuts through the virtual endpoints.

Both keys or labels of a request assigned to a top-marked node $v$
lie in its placing-node interval, hence in $v$'s interval. Thus
\begin{equation}
 |\beta_q(z(o))-\beta_q(y(o))|\le\operatorname{span}_q(v).
 \label{acc:replacement-span}
\end{equation}
The triangle inequality in Lemma~\ref{acc:global-accounting}
shows that charges bound crossing increases.
To bound the charges themselves, we combine the definition of
$\chi_q$, \eqref{acc:replacement-span}, and
Lemma~\ref{acc:topmarks} to obtain
\begin{equation}
 \Gamma_q(v)\le2\capc(v)\operatorname{span}_q(v).
 \label{acc:component-charge}
\end{equation}

To bound their total $D_{\rm op}(q)$, we now replace the sum over
actual top marks by a sum over overflow witnesses.
Lemma~\ref{acc:topmarks} guarantees that every internal top mark
has a locally overflowing child.

Bucket capacities are $2T=1024$ in growing phases and at most
$B+7\le345$ in nonfinal shrinking phases
(Remark~\ref{alg:uniform-bounds}). Set $\kappa=2048$.
By~\eqref{acc:component-charge}, a selected node, when top-marked,
has charge at most $\kappa$ if it is a leaf and at most
$2(1024\cdot2s)=4096s$ if it is a $2s$-leaf parent.
At each scale $s$, let $v$ be the selected parent with $2s$ leaves
and $u_1,u_2$ its children, each with $s$ leaves. Define
\[
 \mathcal E_s(q)=\{N_{\rm op}(u_1)>\capc(u_1)
                 \text{ or }N_{\rm op}(u_2)>\capc(u_2)\}.
\]
Thus, whenever $v$ is top-marked, $\mathcal E_s(q)$ occurs and
its charge satisfies
$\Gamma_q(v)\le4096s\,\mathbf1[\mathcal E_s(q)]$.
The event may also occur when $v$ is not top-marked, so summing
these witness bounds can only enlarge the total.
Nonselected top marks contribute zero. Summing over all scales
and adding the leaf bound gives
\begin{equation}
 D_{\rm op}(q)\le\kappa+
 4096\sum_{\substack{s\in\{1,2,4,\ldots\}\\2s\le b}}
 s\,\mathbf1[\mathcal E_s(q)].
 \label{acc:charge-majorant}
\end{equation}
Here $b$ is the current number of buckets.
For $n<16T$, sequential placement gives $C(q)\le n-1$.
\subsubsection{Bounding reallocation requests by input counts}
\label{acc:new-reference}\label{sec:reallocation-witnesses}

Reallocation demand counts surviving empty cells. We bound it using
old-label counts after growing phases and preceding canonical-arrival
counts between shrinking phases.

\begin{lemma}[Reallocation requests after a growing phase]\label{acc:old-labels}
Let $P_j$ be the prefix set observed before growing phase $j$ begins.
The phase-$j$ tree has boundaries every $T$ prefix ranks and capacity
$2T$ per bucket. At the following reallocation, let $I$ be a node
interval of any new partition tree and put $Z(I)=|P_j\cap I|$.
The number of old empty-cell requests canonical in $I$ is at most
\begin{equation}\label{acc:old-label-bound}
 2Z(I)+4T.
\end{equation}
\end{lemma}

\phantomsection\label{proof:old-labels}
\begin{proof}
In the algorithm, the next tree is built from $P_{j+1}$, which includes
all arrivals through phase $j$. Thus $P_{j+1}$ determines the new
interval $I$, while the surviving empty cells still carry labels
assigned using $P_j$. A surviving cell's request is canonical in $I$
exactly when its old label lies in $I$.

To count these labels, consider $P_j$ in sorted order. The finite old
labels occur at ranks $T,2T,\ldots$. Since $I$ is an interval, the
$Z(I)$ keys of $P_j$ lying in $I$ form a consecutive segment of this
order. Labels are $T$ ranks apart, so this segment contains at most
$Z(I)/T+1$ finite labels. Including the possible label $-\infty$
gives at most $Z(I)/T+2$ old labels.

Each old label was assigned to $2T$ cells. Counting all these cells,
whether now filled or empty, bounds the surviving requests by
\[
 2T\left(\frac{Z(I)}T+2\right)=2Z(I)+4T.
\]
\end{proof}

This bound holds for any new partition tree, so it applies both
between growing phases and at the first shrinking allocation.
If old empty-cell requests locally overflow $I$ during reallocation,
then $2Z(I)+4T>\capc(I)$.
For the algorithm's new tree, Section~\ref{prob:growing} bounds
this event conditional on $P_{j+1}$, which fixes $I$.
\paragraph*{Adjacent-pair merging.}
\phantomsection\label{acc:frozen-reuse}
In subsequent shrinking reallocations, each old label's canonical
destination is its old bucket's parent. Excess empty-cell demand
implies too few canonical arrivals in the preceding phase.

\begin{lemma}[Reallocation requests between shrinking phases]
\label{acc:underflow}
Consider the reallocation at the start of shrinking phase
$2\le i\le k$. Let $I$ be a current node interval containing $s$
buckets, formed by merging $2s$ buckets from phase $i-1$.
Let $N_{i-1}(I)$ count arrivals canonical in $I$ during phase $i-1$,
and let $N_i^{\mathrm{re}}(I)$ count the surviving empty cells
whose reallocation requests are canonical in $I$. Then
\begin{equation}\label{acc:empty-bound}
 N_i^{\mathrm{re}}(I)
 \le\max\{2(B+7)s-N_{i-1}(I),0\}.
\end{equation}
If $i<k$ and these reallocation requests locally overflow $I$, then
\begin{equation}\label{acc:underflow-witness}
 N_{i-1}(I)<(B+14)s\le\frac{11}{10}Bs.
\end{equation}
\end{lemma}

\phantomsection\label{proof:underflow}
\begin{proof}
At the start of phase $i-1$, the interval $I$ contains $2s$ buckets,
each with capacity at most $B+7$. Thus at most $2(B+7)s$ cells
are available in $I$.

If all $N_{i-1}(I)$ canonical arrivals are placed inside $I$,
they fill that many distinct cells. Arrivals from outside $I$
can only reduce the number of remaining empty cells.
If a canonical arrival is placed outside $I$, the routing rule
requires $I$ to be already full, so no empty cells remain.
Consequently, at most $\max\{2(B+7)s-N_{i-1}(I),0\}$ empty cells remain.

At the start of phase $i$, the $2s$ old buckets merge into $s$
new buckets. The surviving empty cells request their old buckets'
parents, so exactly those remaining in $I$ are counted by
$N_i^{\mathrm{re}}(I)$. This proves~\eqref{acc:empty-bound}.

For $i<k$, each new bucket has capacity at least $B$, giving
$I$ total capacity at least $Bs$. Reallocation overflow therefore
requires $N_i^{\mathrm{re}}(I)>Bs$. The positive remainder in
\eqref{acc:empty-bound} then gives
\[
 Bs<N_i^{\mathrm{re}}(I)
 \le2(B+7)s-N_{i-1}(I).
\]
Rearranging gives $N_{i-1}(I)<(B+14)s$.
Finally, $B\ge167\ge140$ implies $14\le B/10$, proving
\eqref{acc:underflow-witness}.
\end{proof}

The same arrival counts therefore bound both placement charges and
the next nonfinal reallocation charges. The first shrinking reallocation
uses Lemma~\ref{acc:old-labels}, and the final phase uses
Lemma~\ref{acc:global-accounting}.

\subsection{Probabilistic tools}
\label{prob:tools}\label{sec:moment-overview}

Through Section~\ref{prob:assembly}, we analyze independent uniform
inputs in $[0,1]$ and fix a deterministic cut $q\in(0,1)$.
Almost surely, the input keys are distinct and none equals $q$.
All bounds are uniform in $q$ and $n$.

For probability calculations, $|I|$ denotes the length of
$I\cap[0,1]$. Clipping virtual endpoints to $0$ and $1$ preserves
input memberships and accounting spans.
We use binomial counts for fixed intervals under independent sampling,
and hypergeometric counts when, conditional on a larger prefix set,
an earlier prefix is a uniform subset of that set.
\begin{samepage}
\begin{lemma}[Binomial and hypergeometric exponential moments]
\label{prob:chernoff}
\emph{Binomial sampling.}
Let $Z\sim\operatorname{Bin}(\ell,p)$, where $\ell$ is a nonnegative
integer and $p\in[0,1]$, and write $\mu=\E[Z]=\ell p$.
For every $u\in\mathbb R$,
\begin{equation}
 \E[e^{uZ}]\le (1-p+pe^u)^\ell
 \le \exp\bigl(\mu(e^u-1)\bigr).
 \label{prob:sampling-mgf}
\end{equation}
If $\mu>0$, then
\begin{align}
 \Prb(Z\ge a)&\le
 \exp\left(-\frac{(a-\mu)^2}{2a}\right)
 &&(a\ge\mu), \label{prob:upper-tail}\\
 \Prb(Z\le a)&\le
 \exp\left(-\frac{(\mu-a)^2}{2\mu}\right)
 &&(0\le a\le\mu). \label{prob:lower-tail}
\end{align}
When $\mu=0$, $Z=0$ almost surely.

\emph{Sampling without replacement.}
The same bounds hold for
$Z\sim\operatorname{Hypergeom}(P,K,\ell)$:
among $P$ elements, $K$ are marked, and $Z$ counts the marked
elements in a uniformly chosen $\ell$-subset.
Here $P\ge1$ and $0\le K,\ell\le P$ are integers,
$p=K/P$, and $\mu=\E[Z]=\ell K/P$.
\end{lemma}
\end{samepage}

These are standard exponential-moment and Chernoff bounds.
Hoeffding's convex comparison for sampling without replacement
\cite{Hoeffding1963} extends the binomial moment bound to the
hypergeometric case, yielding the same tail bounds.

With weights $h\in\{1,2,4,\ldots\}$, the largest active scale bounds the sum even when
witnesses are dependent.
\begin{samepage}
\begin{lemma}[Witnesses at power-of-two scales]\label{prob:dyadic-mgf}
Let $(\mathcal E_h)_{h\in\{1,2,4,\ldots\}}$ be events with
arbitrary dependence. Suppose that, possibly conditional on a
sigma-field $\mathcal H$,
\[
 \Prb(\mathcal E_h\mid\mathcal H)\le A_*e^{-a_*h}
 \qquad(h\ge h_*),
\]
where $A_*,a_*,h_*>0$ are deterministic constants. For every fixed
$c_*>0$ and $0<\eta<a_*/(2c_*)$,
\begin{equation}\label{prob:dyadic-bound}
 \E\left[\exp\left(\eta c_*\sum_h h\,\mathbf1[\mathcal E_h]\right)
          \middle|\mathcal H\right]\le K_*,
\end{equation}
where $K_*<\infty$ depends only on the displayed constants.
\end{lemma}
\end{samepage}

The tail hypothesis is verified for the algorithm's witnesses in
Sections~\ref{prob:growing} and~\ref{prob:shrinking}.

\phantomsection\label{proof:dyadic-mgf}
\begin{proof}
First retain only finitely many scales. Put
$S_*:=\sum_{h<h_*,\,h\in\{1,2,4,\ldots\}}h$, and let $h_{\max}$ be the largest
scale $h\ge h_*$ for which $\mathcal E_h$ occurs, or zero if there
is none. Split the sum at $h_*$ and bound each small-scale
indicator by one:
\[
 \begin{aligned}
 \sum_h h\,\mathbf1[\mathcal E_h]
 &=\sum_{h<h_*}h\,\mathbf1[\mathcal E_h]
   +\sum_{h\ge h_*}h\,\mathbf1[\mathcal E_h]\\
 &\le S_*+\sum_{h\ge h_*}h\,\mathbf1[\mathcal E_h].
 \end{aligned}
\]
If $h_{\max}>0$, all terms with $h>h_{\max}$ vanish. Including every power-of-two scale up
to $h_{\max}$, even those whose indicators are zero, gives
\[
 \sum_{h\ge h_*}h\,\mathbf1[\mathcal E_h]
 \le\sum_{\substack{h\le h_{\max}\\h\in\{1,2,4,\ldots\}}}h
 \le h_{\max}+\frac{h_{\max}}{2}+\frac{h_{\max}}{4}+\cdots=2h_{\max}.
\]
If $h_{\max}=0$, the large-scale contribution is zero. Thus, in either case,
\[
 \sum_h h\,\mathbf1[\mathcal E_h]\le S_*+2h_{\max}.
\]
Moreover, $h_{\max}=h>0$ implies that $\mathcal E_h$ occurs, giving
$\Prb(h_{\max}=h\mid\mathcal H)\le A_*e^{-a_*h}$.
Using the preceding bound and expanding the conditional expectation
over the possible values of $h_{\max}$, we obtain
\[
\begin{aligned}
 &\E\!\left[
   e^{\eta c_*\sum_h h\,\mathbf1[\mathcal E_h]}
   \mid\mathcal H
 \right]\\
 &\quad\le
 e^{\eta c_*S_*}\E[e^{2\eta c_*h_{\max}}\mid\mathcal H]\\
 &\quad=
 e^{\eta c_*S_*}
 \left(
   \Prb(h_{\max}=0\mid\mathcal H)
   +\sum_{\substack{h\ge h_*\\h\in\{1,2,4,\ldots\}}}
      e^{2\eta c_*h}\Prb(h_{\max}=h\mid\mathcal H)
 \right)\\
 &\quad\le
 e^{\eta c_*S_*}
 \left(
   1+A_*\sum_{\substack{h\ge h_*\\h\in\{1,2,4,\ldots\}}}
      e^{-(a_*-2\eta c_*)h}
 \right).
\end{aligned}
\]
The series converges because $a_*-2\eta c_*>0$.
The resulting bound does not depend on the truncation.
Conditional monotone convergence therefore proves the claim for all scales.
\end{proof}
\subsubsection{Lengths of selected empirical intervals}

Empirical bucket boundaries control observation counts, but interval
lengths remain random. We bound the probability that either child of
the parent selected by $q$ is too long or too short.
The proof first confines the parent near $q$, then detects abnormal
child lengths using fixed grid intervals.

\begin{samepage}
\begin{lemma}[Child interval lengths from an i.i.d. prefix]
\label{prob:iid-selected}
Let an observed prefix of $m$ independent uniform samples in $[0,1]$
define an empirical partition tree with $b\le m$ leaves, where $b$
is a power of two.
Fix a deterministic cut $q\in(0,1)$ and a level with $2s$ leaves
per parent. Let $J_q$ be the parent containing $q$.
Each of its two children contains $s$ consecutive buckets, with
nominal length $p=s/b$ and nominal sample count $r=mp$.
If $r\ge1024$, then
\begin{equation}
 \Prb\left(\exists\text{ child }I\text{ of }J_q:
       |I|\notin[15p/16,17p/16]\right)
 \le A_p e^{-a_pr},
 \label{prob:iid-selected-tail}
\end{equation}
where one may take
\[
 A_p=2+2(8195)^2,\qquad a_p=\frac1{14336}.
\]
The constants are uniform in $q,m,b,s$, including cuts outside the
observed range of the prefix.
\end{lemma}
\end{samepage}

\begin{proof}
Let $Z(I)$ count observations from this same prefix lying in $I$.
Throughout the proof, intersect node intervals with $[0,1]$.
Because the buckets are defined by rank, a child's $s$ consecutive
leaves have boundary-rank difference $ms/b$, up to rounding and
endpoint conventions. Thus
\begin{equation}
 r-2\le Z(I)\le r+2,\qquad Z(J_q)\le2r+2.
 \label{prob:rank-counts}
\end{equation}
For every prefix, each child's count is within two of $r$,
although its length remains random. Fixed intervals will detect
when this count is unusually small or large for that length.

First, restrict the parent to a fixed neighborhood of $q$.
Set $\varepsilon=1/32$ and $W=[q-16p,q+16p]\cap[0,1]$.
If $J_q\nsubseteq W$, it contains one of the fixed intervals
$H_-=[q-16p,q]$ and $H_+=[q,q+16p]$.
Only directions with $H_\pm\subseteq[0,1]$ can cause escape.
For either valid interval $H$,
\[
 Z(H)\sim\operatorname{Bin}(m,16p),\qquad
 \E[Z(H)]=16mp=16r.
\]
Yet containment in $J_q$ forces $Z(H)\le2r+2$.
Thus escape requires far fewer prefix observations in this fixed
interval than its mean predicts. A union bound over the two directions
and~\eqref{prob:lower-tail} give
\begin{equation}
 \begin{aligned}
 \Prb(J_q\nsubseteq W)
 &\le\sum_{\substack{H\in\{H_-,H_+\}\\H\subseteq[0,1]}}
             \Prb(Z(H)\le2r+2)\\
 &\le2\exp\left(-\frac{(16r-(2r+2))^2}{2\cdot16r}\right)
 \le2e^{-3r}.
 \end{aligned}
 \label{prob:iid-localization}
\end{equation}

Next, place a fixed grid on $W$ with spacing
$\delta=\varepsilon p/8=p/256$, including both endpoints.
Since $|W|\le32p$, at most $32p/\delta+3=8195$ grid points suffice.
The grid depends only on $q$ and $p$, not on the observed prefix.
Consequently, for every fixed interval $H$ with grid endpoints,
\[
 Z(H)\sim\operatorname{Bin}(m,|H|),\qquad
 \E[Z(H)]=m|H|.
\]
We will round each abnormal child to a grid interval.
Although the interval chosen by rounding depends on the prefix,
all candidates belong to this fixed family. We bound each candidate
and then take a union bound, without conditioning on $J_q\subseteq W$.

Suppose $J_q\subseteq W$ and a child $I$ has length greater than
$(1+\varepsilon)p$. Round its endpoints inward to grid points lying in $I$; if an excluded endpoint is a grid point, use the next grid point inward. This gives a closed grid
interval $H\subseteq I$, removing at most $2\delta=\varepsilon p/4$ of length.
Both counts $Z(H)$ and $Z(I)$ use the same prefix. Thus containment
and the rank bound give
\[
 |H|\ge |I|-2\delta>(1+3\varepsilon/4)p,\qquad
 Z(H)\le Z(I)\le r+2.
\]
Thus a child that is too long produces a long grid interval containing
too few observations. For any fixed grid interval of this length,
the mean is at least $\mu_0=(1+3\varepsilon/4)r$.
Since $\varepsilon r\ge32$ and $\varepsilon\le1$,
\[
 \mu_0-(r+2)=\frac{3\varepsilon r}{4}-2
       \ge\frac{\varepsilon r}{2},\qquad
 \mu_0=\left(1+\frac{3\varepsilon}{4}\right)r\le\frac{7r}{4}.
\]
The lower-tail exponent increases with the mean above the threshold, so
\begin{equation}
 \begin{aligned}
 \Prb(Z(H)\le r+2)
 &\le\exp\left(-\frac{(\mu_0-(r+2))^2}{2\mu_0}\right)\\
 &\le\exp\left(-\frac{(\varepsilon r/2)^2}{2(7r/4)}\right)
 =e^{-\varepsilon^2r/14}.
 \end{aligned}
 \label{prob:iid-long}
\end{equation}

If instead $|I|<(1-\varepsilon)p$, round outward to obtain
a grid interval $H\supseteq I$:
\[
 |H|\le |I|+2\delta<(1-3\varepsilon/4)p,\qquad
 Z(H)\ge Z(I)\ge r-2.
\]
Thus a child that is too short produces a short grid interval
containing too many observations from the same prefix.
For any fixed grid interval of this length, its mean
$\mu=m|H|\le(1-3\varepsilon/4)r$ satisfies
\[
 (r-2)-\mu\ge\frac{3\varepsilon r}{4}-2
       \ge\frac{\varepsilon r}{2},\qquad r-2\le r.
\]
The upper-tail bound therefore gives
\begin{equation}
 \begin{aligned}
 \Prb(Z(H)\ge r-2)
 &\le\exp\left(-\frac{((r-2)-\mu)^2}{2(r-2)}\right)\\
 &\le\exp\left(-\frac{(\varepsilon r/2)^2}{2r}\right)
 =e^{-\varepsilon^2r/8}.
 \end{aligned}
 \label{prob:iid-short}
\end{equation}

Finally, two grid endpoints determine $H$, so there are at most
$8195^2$ candidate intervals. If either child has abnormal length,
then either $J_q$ escapes $W$, or one of the fixed grid intervals
has the low or high count just described. A union bound over these
three possibilities gives
\[
 \begin{aligned}
 &\Prb\left(\exists\text{ child }I\text{ of }J_q:
       |I|\notin[(1-\varepsilon)p,(1+\varepsilon)p]\right)\\
 &\quad\le2e^{-3r}
       +8195^2e^{-\varepsilon^2r/14}
       +8195^2e^{-\varepsilon^2r/8}\\
 &\quad\le(2+2\cdot8195^2)e^{-r/14336}=A_pe^{-a_pr}.
 \end{aligned}
\]
Since $[(1-\varepsilon)p,(1+\varepsilon)p]\subseteq[15p/16,17p/16]$,
this also bounds the event in~\eqref{prob:iid-selected-tail}.
Clipping $W$ and retaining its endpoints covers outer cuts as well.
\end{proof}
\subsection{Growing phases by reverse exposure}
\label{prob:growing}\label{sec:growing-moments}

Assume $n\ge8192$. Recall $t_j=T2^j$ and $L=J+1$.
Let $P_j$ be the unordered set of the first $t_j$ observations,
$0\le j\le L$, so $|P_L|=m=2t_J$.
Growing phase $j$ uses the ranks in $P_j$ to define its bucket
boundaries, then places the next $t_j$ arrivals,
$P_{j+1}\setminus P_j$. The following allocation uses the enlarged
prefix $P_{j+1}$.

For the analysis, we reverse this order of exposure: fix the
unordered set $P_{j+1}$ while leaving its arrival order random.
Conditional on $P_{j+1}$, the earlier prefix $P_j$ is a uniform
$t_j$-subset of its $2t_j$ elements, and the phase-$j$ arrivals
form the complementary half.

\subsubsection{Half-sample overflow}

We first bound the number of new arrivals canonical in a selected
child interval. The interval depends on $P_j$, so we again use
fixed grid intervals to detect an abnormal empirical interval.
Here the relevant counts arise from sampling without replacement.

\begin{samepage}
\begin{lemma}[Half-sample overflow]
\label{prob:half-selected}
Fix a growing phase $j$, its strictly ordered enlarged prefix
$P_{j+1}$ of $2t_j$ elements, and a strict cut in this set, allowing
either outer cut. Conditional on $P_{j+1}$, the earlier prefix $P_j$
is a uniform $t_j$-subset. The empirical partition tree constructed
from $P_j$ has $g_j=t_j/T$ leaves, and $P_{j+1}\setminus P_j$
comprises the arrivals of phase $j$.
At a level with $2s$ leaves per parent, each child of the selected
parent has $s$ buckets and capacity $2Ts$. Then
\begin{equation}
 \Prb\left(\exists\text{ child }I:
       |(P_{j+1}\setminus P_j)\cap I|>2Ts
       \,\middle|\,P_{j+1}\right)
 \le A_ge^{-Ts/224},\qquad A_g=2+1027^2.
 \label{prob:half-overflow}
\end{equation}
The bound is independent of the values in $P_{j+1}$ and of the cut.
\end{lemma}
\end{samepage}

\begin{proof}
Fix $P_{j+1}$ and map its ordered elements to $i/(2t_j+1)$,
$1\le i\le2t_j$. This preserves partition memberships and all
overflow events. Throughout this proof, retain the notation
$P_j,P_{j+1}$ for the mapped sets and intersect node intervals
with $[0,1]$. Thus lengths are measured in normalized rank
coordinates, not in the original numerical values.
Choose a fixed $q'\in(0,1)$ in the gap representing the cut,
including an outer gap when appropriate.

For an interval $H\subseteq[0,1]$, define
\[
 K(H)=|P_{j+1}\cap H|,\qquad Z(H)=|P_j\cap H|.
\]
The first count uses the entire enlarged prefix, while the second
uses the earlier prefix that defines the empirical tree.
Their difference counts the phase-$j$ arrivals in $H$.

For each fixed $H$, the population count $K(H)$ is fixed
conditional on $P_{j+1}$. Among the $2t_j$ elements of $P_{j+1}$,
exactly $K(H)$ lie in $H$. The earlier prefix $P_j$ is a uniformly
chosen $t_j$-subset of this population. Therefore
\[
 Z(H)\mid P_{j+1}
 \sim\operatorname{Hypergeom}(2t_j,K(H),t_j),
 \qquad
 \E[Z(H)\mid P_{j+1}]
 =t_j\frac{K(H)}{2t_j}=\frac{K(H)}2.
\]
Counting the equally spaced population points incurs an error of
at most one at each endpoint, so
$|K(H)-(2t_j+1)|H||\le2$. Since $|H|\le1$,
\begin{equation}
 \begin{aligned}
 \bigl|K(H)-2t_j|H|\bigr|
 &\le\bigl|K(H)-(2t_j+1)|H|\bigr|+|H|
 \le2+1=3,\\
 \bigl|\E[Z(H)\mid P_{j+1}]-t_j|H|\bigr|
 &=\frac12\bigl|K(H)-2t_j|H|\bigr|\le\frac32.
 \end{aligned}
 \label{prob:rank-discrepancy}
\end{equation}

Put $p=s/g_j$ and $r=t_jp=Ts\ge512$, and let $J_{q'}$
be the selected parent. Its children are defined by ranks in $P_j$,
so~\eqref{prob:rank-counts} gives
\[
 r-2\le Z(I)\le r+2,\qquad Z(J_{q'})\le2r+2.
\]
We will show that neither child is likely to be too long.
Together with its earlier-observation count, this bounds how many
complementary arrivals it can contain.

First localize the parent.
Set $W=[q'-16p,q'+16p]\cap[0,1]$.
If $J_{q'}\nsubseteq W$, it contains a fixed interval $H$ of
length $16p$ immediately to the left or right of $q'$.
Only directions for which this interval lies in $[0,1]$ can
cause escape. Containment forces
$Z(H)\le Z(J_{q'})\le2r+2$, whereas its conditional mean is at
least $16r-2$ by~\eqref{prob:rank-discrepancy}.
The lower-tail bound and a union bound over the two directions give
\[
 \Prb(J_{q'}\nsubseteq W\mid P_{j+1})
 \le2\exp\left(-\frac{(14r-4)^2}{2(16r-2)}\right)
 \le2e^{-3r}.
\]

Next, place a fixed grid on $W$ with spacing $p/32$, including
its endpoints. Since $|W|\le32p$, at most
$32p/(p/32)+3=1027$ grid points suffice.
Conditional on $P_{j+1}$, this grid is fixed and does not depend
on which half forms $P_j$. For each fixed grid interval $H$,
the count $Z(H)$ therefore has the hypergeometric distribution
described above.

Suppose $J_{q'}\subseteq W$ and a child $I$ has length greater
than $5p/4$. Round its endpoints inward to grid points lying in $I$; if an excluded endpoint is a grid point, use the next grid point inward. This gives a closed grid interval
$H\subseteq I$, removing at most $2(p/32)=p/16$ of length.
Both counts use the same earlier prefix, so
\[
 |H|\ge |I|-\frac{p}{16}>\frac{19p}{16},
 \qquad Z(H)\le Z(I)\le r+2.
\]
Although the interval selected by rounding depends on $P_j$,
it belongs to the fixed grid family. We bound the low-count event
for each sufficiently long fixed candidate and then take a union
bound.

For each such fixed interval $H$,
equation~\eqref{prob:rank-discrepancy} bounds the conditional mean
of $Z(H)$ from below by $\mu_0=19r/16-2$.
Since $r\ge512$,
\[
 \mu_0-(r+2)=\frac{3r}{16}-4\ge\frac r8,
 \qquad \mu_0\le\frac{7r}{4}.
\]
Applying the lower-tail bound at this minimum mean gives
\[
 \Prb(Z(H)\le r+2\mid P_{j+1})
 \le\exp\left(-\frac{(r/8)^2}{2\cdot7r/4}\right)
 =e^{-r/224}.
\]
A union bound over at most $1027^2$ fixed grid intervals and
the escape event shows that both child lengths are at most
$5p/4$, except with probability at most
\[
 2e^{-3r}+1027^2e^{-r/224}\le A_ge^{-r/224}.
\]
These tail bounds condition only on $P_{j+1}$, not on the
selected child's identity or on $J_{q'}\subseteq W$.

Finally, consider any child satisfying $|I|\le5p/4$.
The population discrepancy gives $K(I)\le2t_j|I|+3$,
while the empirical rank bound gives $Z(I)\ge r-2$.
Both bounds hold for every realization, including this
sample-dependent interval. Subtracting the earlier observations
from the full population therefore gives
\[
 \begin{aligned}
 |(P_{j+1}\setminus P_j)\cap I|
 &=K(I)-Z(I)\\
 &\le2t_j\frac{5p}{4}+3-(r-2)\\
 &=\frac32r+5<2r=2Ts.
 \end{aligned}
\]
Thus neither child locally overflows outside the exceptional
event above. Since $r=Ts$, this proves~\eqref{prob:half-overflow}.
\end{proof}

\subsubsection{Reallocation after a half-split}

After growing phase $j$, surviving empty cells retain labels from
$P_j$, but reallocation uses the new tree built from $P_{j+1}$.
We bound their requests to a child interval of this new tree.

For a new node interval $I$, put $Z(I)=|P_j\cap I|$, the number
of earlier observations in $I$. Counting all cells assigned to old
labels in $I$, Lemma~\ref{acc:old-labels} gives
\begin{equation}
 \text{old empty-cell requests canonical in }I
 \le 2Z(I)+4T.
 \label{prob:all-labels}
\end{equation}

Conditional on $P_{j+1}$, the new tree and the two children
of its selected parent are fixed, while $P_j$ remains
a uniform half-subset.
For each such $I$,
\[
 Z(I)\mid P_{j+1}
 \sim\operatorname{Hypergeom}
       \bigl(2t_j,|P_{j+1}\cap I|,t_j\bigr),
 \qquad
 \E[Z(I)\mid P_{j+1}]=\frac{|P_{j+1}\cap I|}{2}.
\]
Unlike in Lemma~\ref{prob:half-selected}, the interval being tested
is fixed by the conditioning, so we can apply the sampling tail
bound directly, without a grid.

\begin{samepage}
\begin{lemma}[Reallocation after a half-split]
\label{prob:half-reallocation}
Conditional on the enlarged prefix set $P_{j+1}$, consider the two
children $I$ of the selected $2s$-leaf parent in the new tree.
If $j<J$, the new tree is a growing tree, and each child has
$s$ buckets of capacity $2T$. Then
\begin{equation}
 \Prb\bigl(\exists I:2Z(I)+4T>2Ts\mid P_{j+1}\bigr)
 \le2e^{-Ts/96}\qquad(s\ge8).
 \label{prob:growing-reallocation}
\end{equation}
If $j=J$, the new tree is the first shrinking tree, with
$b=b_1$ finest leaves, and each child has capacity at least $Bs$.
Then
\begin{equation}
 \Prb\bigl(\exists I:2Z(I)+4T>Bs\mid P_L\bigr)
 \le2e^{-Bs/100}\qquad(s\ge64),
 \label{prob:first-frozen-reallocation}
\end{equation}
where $P_j=P_J$ and $P_{j+1}=P_L$.
Each event above contains the corresponding local reallocation-overflow
event.
\end{lemma}
\end{samepage}

\begin{proof}
Fix $P_{j+1}$ and one of the two new children $I$, and write
$\mu=\E[Z(I)\mid P_{j+1}]$.
We bound the probability that $2Z(I)+4T$ exceeds $2Ts$
in the growing case or $Bs$ in the first shrinking case.

For a subsequent growing tree, the $s$ consecutive buckets in $I$
contain at most $Ts+2$ observations from $P_{j+1}$ by
\eqref{prob:rank-counts}. Thus, for $s\ge8$,
\[
 \begin{aligned}
 \mu&\le\frac{Ts+2}{2}
       =\frac{Ts}{2}+1\le\frac{5Ts}{8},\\
 2Z(I)+4T>2Ts
 &\implies Z(I)>(s-2)T\ge\frac{3Ts}{4}.
 \end{aligned}
\]
The threshold $3Ts/4$ exceeds this mean, so
\eqref{prob:upper-tail} gives
\[
 \begin{aligned}
 &\Prb\bigl(2Z(I)+4T>2Ts\mid P_{j+1}\bigr)\\
 &\quad\le
 \exp\left(-\frac{(3Ts/4-5Ts/8)^2}{2\cdot3Ts/4}\right)
 =e^{-Ts/96}.
 \end{aligned}
\]

For the first shrinking tree, a child with $s$ finest leaves
contains at most $ms/b+2$ observations from $P_L$.
Here $m\le N/3$, $b=2^{k-1}$, and
$\lambda=4N/(3\cdot2^k)$.
Since $B=\lfloor\lambda\rfloor-3$, we have $\lambda<B+4$.
Consequently,
\[
 \mu\le\frac{ms}{2b}+1
 \le\frac{\lambda s}{4}+1
 \le\frac{(B+4)s}{4}+1
 <\frac{13}{50}Bs.
\]
The last inequality uses $B\ge167$ and $s\ge64$.
These bounds also give $4T\le Bs/4$, so
\[
 2Z(I)+4T>Bs
 \implies Z(I)>\frac{Bs-4T}{2}\ge\frac{3Bs}{8}.
\]
Applying~\eqref{prob:upper-tail} at threshold $3Bs/8$ gives
\[
 \begin{aligned}
 &\Prb\bigl(2Z(I)+4T>Bs\mid P_L\bigr)\\
 &\quad\le
 \exp\left(-\frac{(3/8-13/50)^2Bs}{2\cdot3/8}\right)
 \le e^{-Bs/100}.
 \end{aligned}
\]
Taking a union bound over the two children proves both claims.
\end{proof}

\subsubsection{The reverse conditional product}

For a fixed cut $q$, we pair the arrival charge in growing phase $j$
with the reallocation charge immediately after that phase.
Arrivals are placed using the tree built from $P_j$, whereas the
following reallocation uses the new tree built from $P_{j+1}$.
We first bound these two charges using their overflow witnesses,
then combine their exponential moments across growing phases.

For each scale $s\in\{1,2,4,\ldots\}$, let
$\mathcal E^{\mathrm{arr}}_{j,s}$ be the event that a child $I$ of the selected
$2s$-leaf parent in the current tree satisfies
\[
 |(P_{j+1}\setminus P_j)\cap I|>2Ts.
\]
Each child has capacity $2Ts$, so this event is local arrival
overflow in at least one child.

Let $\mathcal E^{\mathrm{re}}_{j,s}$ be the event that a child $I$ of the
selected $2s$-leaf parent in the next tree satisfies
\[
 \begin{aligned}
 2|P_j\cap I|+4T&>2Ts &&(j<J),\\
 2|P_j\cap I|+4T&>Bs  &&(j=J).
 \end{aligned}
\]
The next tree is a growing tree for $j<J$ and the first shrinking
tree for $j=J$. The corresponding child capacity is $2Ts$ in the
first case and at least $Bs$ in the second.
By~\eqref{prob:all-labels}, local reallocation overflow
implies $\mathcal E^{\mathrm{re}}_{j,s}$, although the converse need not hold.
If the selected parent does not exist in the relevant tree, we take
the corresponding event to be empty.

In~\eqref{acc:charge-majorant}, bound every indicator at scale
$s<64$ by one and combine these contributions with the selected-leaf
bound in a constant $D_g$. Set
\begin{equation}
 \begin{split}
 D_g&=2048+4096(1+2+4+8+16+32)=260096,\\
 A_j^{\mathrm{arr}}
 &=D_g+4096\sum_{s\in\{64,128,256,\ldots\}}
                   s\,\mathbf1[\mathcal E^{\mathrm{arr}}_{j,s}],\\
 R_j
 &=D_g+4096\sum_{s\in\{64,128,256,\ldots\}}
                   s\,\mathbf1[\mathcal E^{\mathrm{re}}_{j,s}],\\
 G_j&=A_j^{\mathrm{arr}}+R_j.
 \end{split}
 \label{prob:growing-majorants}
\end{equation}
By~\eqref{acc:charge-majorant}, $A_j^{\mathrm{arr}}$ and $R_j$
bound the phase-$j$ arrival charge and the following reallocation
charge, respectively.
Thus $G_j$ bounds their sum, including the first shrinking
reallocation when $j=J$.
For fixed $n$ and $q$, these majorants depend only on
$(P_j,P_{j+1})$, not on the within-phase arrival order or surviving
empty cells. This is what permits the reverse-exposure argument below.

\begin{samepage}
\begin{lemma}[Uniform reverse moment]
\label{prob:reverse-mgf}
Set
\begin{equation}
 \theta_g=\frac1{32768},\qquad
 M_g=e^{D_g/16384}
       \left(1+\frac{A_ge^{-32}}{1-e^{-1/2}}\right).
 \label{prob:growing-constants}
\end{equation}
For every fixed prefix set $P_L$ and every strict cut induced in $P_L$,
\begin{equation}
 \E\left[\exp\left(\theta_g\sum_{j=0}^{L-1}G_j\right)
                 \middle|P_L\right]\le M_g^L.
 \label{prob:reverse-product}
\end{equation}
For every deterministic value cut $q$, this conditional inequality
holds almost surely in the i.i.d.\ experiment.
\end{lemma}
\end{samepage}

\begin{proof}
We first bound the exponential moments of
$A_j^{\mathrm{arr}}$ and $R_j$ separately.
For $s\ge64$, Lemmas~\ref{prob:half-selected}
and~\ref{prob:half-reallocation} give
\[
 \begin{aligned}
 \Prb(\mathcal E^{\mathrm{arr}}_{j,s}\mid P_{j+1})
 &\le A_ge^{-s},\\
 \Prb(\mathcal E^{\mathrm{re}}_{j,s}\mid P_{j+1})
 &\le A_ge^{-s}.
 \end{aligned}
\]
Here $T/224,T/96,B/100>1$, and $A_g\ge2$ covers the
prefactors in both reallocation bounds.

For either $U=A_j^{\mathrm{arr}}$ or $U=R_j$, let $s\ge64$
be its largest active scale. The maximum-scale argument of
Lemma~\ref{prob:dyadic-mgf} gives $U-D_g\le4096\cdot2s$,
so its scale-dependent exponential factor at parameter $1/16384$
is at most $e^{s/2}$. The conditional probability of this
largest scale is at most $A_ge^{-s}$.
Including the case with no active scale and extending this
sum to all integers $s\ge64$ gives
\[
 \begin{aligned}
 \E[e^{U/16384}\mid P_{j+1}]
 &\le e^{D_g/16384}
       \left(1+A_g\sum_{s=64}^{\infty}e^{s/2}e^{-s}\right)\\
 &=e^{D_g/16384}
       \left(1+\frac{A_ge^{-32}}{1-e^{-1/2}}\right)
 =M_g.
 \end{aligned}
\]

The two majorants need not be independent. Since
$2\theta_g=1/16384$, the conditional Cauchy--Schwarz inequality
combines their moment bounds:
\begin{equation}
 \begin{split}
 \E[e^{\theta_gG_j}\mid P_{j+1}]
 &\le\Bigl(
   \E[e^{2\theta_g A_j^{\mathrm{arr}}}\mid P_{j+1}]
   \E[e^{2\theta_g R_j}\mid P_{j+1}]
   \Bigr)^{1/2}\\
 &\le M_g.
 \end{split}
 \label{prob:one-split-mgf}
\end{equation}

We now combine the bounds across growing phases.
Let $\mathcal H_j=\sigma(P_j,\ldots,P_L)$.
Conditioning on $\mathcal H_{j+1}$ reveals the larger prefix sets,
while the arrival order within $P_{j+1}$ remains uniform.
Hence $P_j$ remains a uniform half-subset of $P_{j+1}$.
Because $G_j$ depends only on $(P_j,P_{j+1})$,
the same argument as in~\eqref{prob:one-split-mgf} gives
\[
 \E[e^{\theta_gG_j}\mid\mathcal H_{j+1}]\le M_g.
\]
In contrast, each later $G_\ell$, with $\ell>j$, is already
determined by the revealed prefix sets and is therefore
$\mathcal H_{j+1}$-measurable.

The tower property lets us average over $G_j$ while keeping
the later factors fixed:
\[
 \begin{aligned}
 \E\!\left[e^{\theta_g\sum_{\ell=j}^{L-1}G_\ell}\mid P_L\right]
 &=\E\!\left[
 e^{\theta_g\sum_{\ell=j+1}^{L-1}G_\ell}
 \E[e^{\theta_gG_j}\mid\mathcal H_{j+1}]
 \middle|P_L\right]\\
 &\le M_g\,
 \E\!\left[e^{\theta_g\sum_{\ell=j+1}^{L-1}G_\ell}\mid P_L\right].
 \end{aligned}
\]
Repeating this step for $j=0,\ldots,L-1$ contributes one factor
$M_g$ per phase and proves~\eqref{prob:reverse-product}.

The one-split bounds hold for every fixed population and strict cut,
including outer cuts, so they apply almost surely to the cut
induced by deterministic $q$ in $P_L$.
\end{proof}
\subsection{The fixed tree: persistent errors and future blocks}
\label{prob:shrinking}\label{sec:shrinking-moments}

The prefix $P_L$ observed by the end of the growing phases defines
the first shrinking tree, with $b=b_1$ initial buckets, the finest
leaves of the fixed tree.
Subsequent phases merge adjacent buckets within this tree.
In phase $i$, each current bucket combines $2^{i-1}$ initial buckets,
so the current tree has $b_i=b/2^{i-1}$ leaves.

We first bound an exponential moment of the accumulated charge
from initial interval-length errors, averaging over the random
prefix $P_L$. We then condition on $P_L$ to bound the remaining
charges using independent future input blocks.

\subsubsection{Persistent partition errors}

While a node remains internal, its interval and the intervals of
its two children remain unchanged, although they contain fewer
current buckets as phases advance. Thus an initial error in interval
length can affect several phases.

We apply Lemma~\ref{prob:iid-selected} to the initial intervals
and bound the charges their errors produce through later merges.
For each scale $\nu\in\{2,4,\ldots,b\}$, let $J_{q,\nu}$ be the selected
parent containing $q$ and consisting of $\nu$ initial buckets.
Each of its two children contains $\nu/2$ initial buckets.
This is the setting of Lemma~\ref{prob:iid-selected} with
$s=\nu/2$, so the nominal child length and sample count are
\[
 p=\frac{\nu}{2b},
 \qquad
 r=mp=\frac{m\nu}{2b},
 \qquad m=|P_L|.
\]
For $\nu\ge64$, define the same length-error event as in that lemma:
\begin{equation}
 \mathrm{Bad}_\nu
 =
 \left\{
 \exists\text{ child }I\text{ of }J_{q,\nu}:
 |I|\notin[15p/16,17p/16]
 \right\}.
 \label{prob:bad-pilot-definition}
\end{equation}
Substituting $p=\nu/(2b)$ gives the permitted length range
$[15\nu/(32b),17\nu/(32b)]$.

Equation~\eqref{alg:pilot-density} gives $m/b>32$. Hence
\[
 r=\frac{m\nu}{2b}>16\nu\ge1024
 \qquad(\nu\ge64),
\]
which verifies the hypothesis of Lemma~\ref{prob:iid-selected}.
Since $a_p=1/14336$, that lemma gives
\begin{equation}
 \Prb(\mathrm{Bad}_\nu)
 \le A_pe^{-a_pr}
 =A_pe^{-r/14336}
 \le A_pe^{-\nu/896}.
 \label{prob:bad-pilot-tail}
\end{equation}
This probability is taken over the random prefix $P_L$, before
conditioning on its observed values.
For scales $\nu<64$, we declare $\mathrm{Bad}_\nu$ to hold
unconditionally and bound its contribution deterministically.

To bound their accumulated charges across phases, define
\begin{equation}\label{prob:pilot-penalty}
 D_{\rm err}
 =4\kappa\sum_{\nu\in\{2,4,\ldots,b\}}
          \nu\,\mathbf1[\mathrm{Bad}_\nu].
\end{equation}
The proof below bounds each bad parent's cumulative charge by
$4\kappa\nu$. For fixed $n,q$, $D_{\rm err}$ depends only on $P_L$,
not on future arrivals.

\begin{samepage}
\begin{lemma}[Geometrically weighted fixed-partition errors]
\label{prob:pilot-mgf}
The variable $D_{\rm err}$ in~\eqref{prob:pilot-penalty}
dominates all bad-parent charges in nonfinal
shrinking arrivals and subsequent nonfinal bucket-merging reallocations.
There are absolute constants $\theta_b>0$ and $M_b\ge1$ such that,
uniformly in $q$ and $n$,
\begin{equation}
 \E[e^{\theta_bD_{\rm err}}]\le M_b.
 \label{prob:pilot-moment}
\end{equation}
\end{lemma}
\end{samepage}

\begin{proof}
Fix a realization of $P_L$ and any continuation.
Consider a selected bad parent containing $\nu$ initial buckets.
In shrinking phase $i$, each current bucket combines $2^{i-1}$
initial buckets.
Whenever this parent remains internal in the current tree, it
therefore contains $\nu/2^{i-1}$ current buckets.
By~\eqref{acc:component-charge}, its charge is at most
$\kappa\nu/2^{i-1}$ per arrival or reallocation operation.

There are at most two such operations per phase.
Extending the sum over contributing phases to all $i\ge1$ bounds
its cumulative charge by
\[
 2\kappa\nu\sum_{i\ge1}2^{-(i-1)}
 =4\kappa\nu.
\]
Summing this bound over the selected bad parents gives
$D_{\rm err}$ and proves the domination.
The first shrinking reallocation is already included in $G_J$.
Selected-leaf charges and the final phase are covered by the baseline.

We now average over $P_L$. By~\eqref{prob:bad-pilot-tail},
Lemma~\ref{prob:dyadic-mgf} applies with
\[
 (c_*,h_*,A_*,a_*)
 =(4\kappa,64,A_p,1/896)
\]
and therefore proves~\eqref{prob:pilot-moment}.
Unused scales are assigned empty events.
Here the scales start at $\nu=2$, so the small-scale sum,
before multiplication by $4\kappa$, is at most
$2+4+8+16+32=62$.
No independence between the bad events is needed.
\end{proof}
\subsubsection{Good intervals and future arrivals}

Condition on $P_L$. The fixed-tree intervals and their bad events
are now determined, while later inputs remain independent uniform
samples.
Recall from Lemma~\ref{acc:underflow} that $N_i(I)$ counts arrivals
canonical in $I$ during shrinking phase $i$.
Writing $X_t$ for arrival $t$, this is the number of phase-$i$
inputs $X_t$ that fall in $I$, regardless of their actual
destination cells.

Shrinking phase $i$ contains $n_i$ arrivals. Its length is determined
by $n$, not by placements or surviving empty cells.
For any interval $I$ determined by $P_L$,
\[
 N_i(I)\mid P_L\sim\operatorname{Bin}(n_i,|I|),
 \qquad
 \E[N_i(I)\mid P_L]=n_i|I|.
\]
We use this same count to control two overflow mechanisms.
Local arrival overflow requires too many arrivals in a current
interval, whereas local overflow in the following reallocation
requires too few arrivals in a next-tree interval.

A selected parent with $2s$ current buckets in phase $i$ contains
$2^i s$ initial buckets, so it is $J_{q,2^i s}$ in the fixed tree.
We call this parent good when $\mathrm{Bad}_{2^i s}$ does not occur.
By~\eqref{prob:bad-pilot-definition}, each of its two children has
length in
\[
 \left[\frac{15s}{16b_i},\frac{17s}{16b_i}\right].
\]
Thus the length bounds from the preceding subsection directly
control the conditional mean number of arrivals in each child.

\begin{samepage}
\begin{lemma}[Arrival and next-reallocation tails]
\label{prob:good-block-tails}
Condition on $P_L$ and let $1\le i<k$.
A good selected parent $J_{q,2^i s}$ with $2s$ current buckets
in phase $i$ satisfies
\begin{equation}
 \Prb\left(
 \exists\text{ child }I\text{ of }J_{q,2^i s}:
 N_i(I)>Bs
 \,\middle|\,P_L\right)
 \le2e^{-Bs/200}.
 \label{prob:good-arrival}
\end{equation}
If $i+1<k$, a good selected parent $J_{q,2^{i+1}s}$ with
$2s$ current buckets in phase $i+1$ satisfies
\begin{equation}
 \Prb\left(
 \exists\text{ child }I\text{ of }J_{q,2^{i+1}s}:
 N_i(I)<1.1Bs
 \,\middle|\,P_L\right)
 \le2e^{-Bs/65}.
 \label{prob:good-reallocation}
\end{equation}
Each event contains the corresponding local overflow event,
for current arrivals and the following reallocation, respectively.
\end{lemma}
\end{samepage}

\begin{proof}
Fix one child interval $I$ of the parent under consideration and
let $\mu=\E[N_i(I)\mid P_L]=n_i|I|$.

For current arrivals, $I$ contains $s$ phase-$i$ buckets.
Goodness of the parent and~\eqref{alg:arrival-density} give
\[
 \mu
 \le n_i\frac{17s}{16b_i}
 \le\frac{17}{16}\frac34(B+7)s
 \le0.9Bs.
\]
The child's capacity is at least $Bs$, so local arrival overflow
requires $N_i(I)>Bs$.
Applying the upper-tail bound~\eqref{prob:upper-tail} gives
\[
 \Prb(N_i(I)>Bs\mid P_L)
 \le
 \exp\left(-\frac{(Bs-0.9Bs)^2}{2Bs}\right)
 =e^{-Bs/200}.
\]

For the following reallocation, $I$ contains $s$ phase-$(i+1)$
buckets, formed from $2s$ phase-$i$ buckets.
Since $b_{i+1}=b_i/2$, goodness of the next parent gives
\[
 |I|\ge\frac{15s}{16b_{i+1}}
       =\frac{15\cdot2s}{16b_i}.
\]
Together with~\eqref{alg:arrival-density}, this yields
\[
 \mu
 \ge\frac{15}{16}\frac{n_i}{b_i}\,2s
 \ge\frac{15}{16}\left(\frac{3B}{4}-1\right)2s
 \ge1.3Bs.
\]
The numerical comparisons in both cases use $B\ge167$.

Lemma~\ref{acc:underflow}, applied at the start of phase $i+1$,
shows that local reallocation overflow implies
\[
 N_i(I)<(B+14)s\le1.1Bs.
\]
Thus too few arrivals in this interval are necessary for the
following reallocation to overflow.
The lower-tail exponent increases with $\mu$ for $\mu\ge1.1Bs$.
Using $\mu\ge1.3Bs$ in~\eqref{prob:lower-tail} gives
\[
 \Prb(N_i(I)<1.1Bs\mid P_L)
 \le
 \exp\left(-\frac{(1.3Bs-1.1Bs)^2}{2\cdot1.3Bs}\right)
 =e^{-Bs/65}.
\]
Taking a union bound over the two children proves both claims.
\end{proof}

We now convert these events into charge bounds.
By~\eqref{acc:component-charge}, a top-marked parent with $2s$
current buckets contributes at most $2\kappa s$.
As in the growing analysis, we sum this bound over possible
overflow witnesses, including parents that are not top-marked.

For $1\le i<k$, define $W_i$ by the two sums below.
The first bounds good-parent arrival charges in phase $i$,
and the second bounds good-parent reallocation charges at the
start of phase $i+1$. Each indicator tests whether either child
of the parent specified under its sum satisfies the inequality:
\begin{equation}\label{prob:raw-block-majorant}
 \begin{split}
 W_i={}&2\kappa
 \sum_{\substack{2s\le b_i\\
                 s\in\{1,2,4,\ldots\}\\
                 J_{q,2^i s}\text{ good}}}
 s\,\mathbf1\!\left[
   \exists\text{ child }I:N_i(I)>Bs
 \right]\\
 &+\mathbf1[i+1<k]\,2\kappa
 \sum_{\substack{2s\le b_{i+1}\\
                 s\in\{1,2,4,\ldots\}\\
                 J_{q,2^{i+1}s}\text{ good}}}
 s\,\mathbf1\!\left[
   \exists\text{ child }I:N_i(I)<1.1Bs
 \right].
 \end{split}
\end{equation}
The second sum is omitted when the following phase is final.
Bad-parent charges are covered by $D_{\rm err}$, while selected-leaf
charges and the final phase are covered by the baseline.

Both sums use the same phase-$i$ arrival counts, but test intervals
in their respective trees.
Although the actual reallocation depends on earlier placements,
$W_i$ depends only on $P_L$ and the phase-$i$ input block.
We can therefore combine the moment bounds across disjoint future
blocks without conditioning on the surviving empty-cell configurations.

\begin{samepage}
\begin{lemma}[Conditional product for the future blocks]
\label{prob:future-product}
There are absolute constants $\theta_s>0$ and $M_s\ge1$ such that
\begin{equation}
 \E\left[\exp\left(\theta_s\sum_{i=1}^{k-1}W_i\right)
                  \middle|P_L\right]\le M_s^{k-1}.
 \label{prob:shrinking-mgf}
\end{equation}
\end{lemma}
\end{samepage}

\begin{proof}
Given $P_L$, the good parents included in each sum are fixed.
By Lemma~\ref{prob:good-block-tails}, each event appearing at scale
$s$ has conditional probability at most $2e^{-2s/3}$, since
\[
 \frac{Bs}{200},\frac{Bs}{65}\ge\frac{2s}{3}.
\]

Set $\theta_s=1/(24\kappa)$ and consider either sum $U$
in~\eqref{prob:raw-block-majorant}.
If $s$ is its largest active scale, Lemma~\ref{prob:dyadic-mgf}
gives $U\le2\kappa\cdot2s=4\kappa s$.
At parameter $2\theta_s=1/(12\kappa)$, the exponential factor
is at most $e^{s/3}$, and the conditional probability of this
largest scale is at most $2e^{-2s/3}$.
Including the case with no active scale and extending this
sum to all positive integers gives
\[
 \begin{aligned}
 \E[e^{2\theta_s U}\mid P_L]
 &\le1+2\sum_{s=1}^{\infty}e^{-s/3}\\
 &=1+\frac{2e^{-1/3}}{1-e^{-1/3}}\\
 &\le1+\frac{2e^{-1/3}}{1-e^{-1/6}}
 =:M_s.
 \end{aligned}
\]
An absent sum has exponential moment one.

The two sums in $W_i$ use the same input block and need not be
independent. As in~\eqref{prob:one-split-mgf}, conditional
Cauchy--Schwarz combines their bounds at $\theta_s$:
\[
 \E[e^{\theta_sW_i}\mid P_L]\le M_s.
\]

Finally, phase lengths are determined by $n$, so different phases
use disjoint future input blocks.
These blocks remain mutually independent conditional on $P_L$,
and each $W_i$ depends only on its own block and $P_L$.
Their conditional exponential moments therefore multiply:
\[
 \begin{aligned}
 \E\left[\exp\left(\theta_s\sum_{i=1}^{k-1}W_i\right)
          \middle|P_L\right]
 &=\prod_{i=1}^{k-1}
   \E[e^{\theta_sW_i}\mid P_L]\\
 &\le M_s^{k-1}.
 \end{aligned}
\]
This proves~\eqref{prob:shrinking-mgf}.
\end{proof}

\subsection{A fixed-cut exponential moment}
\label{prob:assembly}\label{sec:assembly-main}

We now combine the three charge bounds.
Conditional on $P_L$, the partition error is fixed, while the
growing and future-block sums are independent. We first combine
their conditional moments, then average over $P_L$.

For the partition error, take
\begin{equation}\label{prob:explicit-error-constants}
 \theta_b=\frac{1}{16\kappa\cdot896}=\frac1{29360128},\qquad
 M_b=e^{62/3584}
 \left(1+A_p\frac{e^{-64/1792}}{1-e^{-1/1792}}\right).
\end{equation}
In this application of Lemma~\ref{prob:dyadic-mgf}, scales start
at $\nu=2$, so the small-scale sum is at most $62$.
The large-scale exponent is
$-\nu/896+2\theta_b(4\kappa)\nu=-\nu/1792$.
Summing over integers $\nu\ge64$ gives $M_b$.
The three moment bounds are summarized in
Table~\ref{tab:moment-parameters}. Use the common parameter
\begin{equation}\label{prob:common-explicit-parameter}
 \theta=\min\{\theta_g,\theta_b,\theta_s\}
       =\frac1{29360128}.
\end{equation}
These parameters are independent of the cut, input size, and
failure exponent.

\begin{table}[htbp]
\centering
\caption{Fixed-cut moment bounds, $\kappa=2048$. Here $L,k$ count growing and
shrinking phases. Bounds remain valid at smaller positive parameters.}
\label{tab:moment-parameters}
\begin{tabularx}{\linewidth}{@{}>{\raggedright\arraybackslash}p{0.23\linewidth}>{\raggedright\arraybackslash}p{0.24\linewidth}Y@{}}
\toprule
Charge & Parameter & Moment bound \\
\midrule
Growing sum $\sum_jG_j$ & $\theta_g=1/32768$
& At most $M_g^L$, conditional on $P_L$ \\
\addlinespace
Partition error $D_{\rm err}$ & $\theta_b=1/29360128$
& At most $M_b$, averaged over $P_L$ \\
\addlinespace
Future sum $\sum_iW_i$ & $\theta_s=1/(24\kappa)$
& At most $M_s^{k-1}$, conditional on $P_L$ \\
\bottomrule
\end{tabularx}
\end{table}

\begin{samepage}
\begin{lemma}[Witness decomposition and phase moments]\label{main:phase-bounds}
For $n\ge8192$ independent uniform inputs and a fixed $q\in(0,1)$,
the nonnegative variables $G_j,D_{\rm err},W_i$ from
\eqref{prob:growing-majorants}, \eqref{prob:pilot-penalty}, and
\eqref{prob:raw-block-majorant} satisfy, on the probability-one event
of distinct keys unequal to $q$,
\begin{equation}\label{main:execution-decomposition}
 C(q)\le D_1(L+k+1)+\sum_{j=0}^{L-1}G_j
           +D_{\rm err}+\sum_{i=1}^{k-1}W_i,
 \qquad D_1=7000.
\end{equation}
The moment bounds in Table~\ref{tab:moment-parameters}
hold at the common parameter $\theta$,
almost surely for the conditional bounds.
The growing and future-block sums are independent conditional on $P_L$.
\end{lemma}
\end{samepage}

\phantomsection\label{proof:phase-bounds}
\begin{proof}
The $G_j$ cover growing arrivals and following reallocations,
including the first shrinking reallocation. Thereafter, $D_{\rm err}$
covers bad-parent charges, and $W_i$ covers good-parent phase-$i$
arrivals and the next nonfinal reallocation.
Lemma~\ref{acc:global-accounting} bounds the remaining charges by
\[
 (T-1)+2(L+k)+2\kappa(k-1)+2060\le7000(L+k+1),
\]
proving~\eqref{main:execution-decomposition}.
The moment bounds follow from Lemmas~\ref{prob:reverse-mgf},
\ref{prob:pilot-mgf}, and~\ref{prob:future-product}.
Nonnegativity permits decreasing their parameters to $\theta$.

Given the unordered $P_L$,
the growing sum depends on its uniform ordering, independent of
the later inputs determining the future-block sum.
\end{proof}

\begin{samepage}
\begin{proposition}[Fixed-cut exponential moment]
\label{prob:fixed-cut-mgf}
There are absolute constants $\theta>0$, $A_0\ge1$, and $\gamma_0\ge0$
such that \EmpiricalSort{}, as specified above, satisfies
\begin{equation}
 \E[e^{\theta C(q)}]\le A_0n^{\gamma_0}
 \qquad\text{for every deterministic }q\in(0,1)
 \label{prob:fixed-cut-bound}
\end{equation}
on $n\ge2$ independent uniform inputs in $[0,1]$, uniformly in $q,n$.
\end{proposition}
\end{samepage}

\begin{proof}
\leavevmode\phantomsection\label{proof:fixed-cut}%
For $n\ge8192$, use $\theta$ from
\eqref{prob:common-explicit-parameter}.
Conditional on $P_L$, $D_{\rm err}$ is fixed and the other sums
are independent. Lemma~\ref{main:phase-bounds} gives
\[
 \begin{aligned}
 \E[e^{\theta C(q)}]
 &\le e^{\theta D_1(L+k+1)}
 \E\!\left[e^{\theta D_{\rm err}}
   \E[e^{\theta\sum_jG_j}\mid P_L]
   \E[e^{\theta\sum_iW_i}\mid P_L]\right]\\
 &\le e^{\theta D_1(L+k+1)}M_g^L M_s^{k-1}M_b.
 \end{aligned}
\]
The last step uses the conditional bounds for the two sums and
the unconditional moment of $D_{\rm err}$.
Since $L,k\le\log_2n$, take
\[
 \gamma_0=\frac{2\theta D_1+\log M_g+\log M_s}{\log2},\qquad
 A_0=\max\{M_be^{\theta D_1},e^{8190\theta}\}.
\]
The second term covers $2\le n<8192$, where
$C(q)\le n-1\le8190$.
\end{proof}

\subsection{From value cuts to arbitrary fixed inputs}
\label{sec:transfer}\label{sec:rank-overview}

For fixed strictly ordered keys, let $C_a$ count adjacent output pairs
with one rank at most $a$ and the other greater than $a$, $1\le a<n$.
Comparison invariance makes the output rank permutation depend only
on the arrival comparison pattern and $n$, not numerical gaps.
We first transfer the value-cut bound to these rank cuts, then use
the numerical gaps to recover the actual cost.

\subsubsection{Comparison-invariant conditioning}

A fixed value cut does not have a fixed rank.
To recover rank cut $a$ from Proposition~\ref{prob:fixed-cut-mgf},
we condition on exactly $a$ observations lying below a
deterministic $q$, without changing the arrival-rank permutation.

\begin{samepage}
\begin{lemma}[Transfer at a binomial mode]\label{tr:mode}
Suppose a deterministic comparison-based algorithm satisfies
\[
 \E[e^{\theta C(q)}]\le A_0 n^{\gamma_0}
\]
under i.i.d.\ uniform input, uniformly in deterministic
$q\in(0,1)$ and $n$. Then, for every fixed strict input exposed
in a uniform random permutation,
\[
 \E_\pi[e^{\theta C_a}]\le2A_0n^{\gamma_0+1}
 \qquad(1\le a<n).
\]
\end{lemma}
\end{samepage}

\begin{proof}
Fix $q=a/n$ in the uniform experiment and let
$N_q=\sum_{t=1}^n\mathbf1[X_t\le q]$.
On $N_q=a$, the value cut $q$ is the rank cut $a$, so $C(q)=C_a$.
The arrival-rank permutation is independent of the sorted values,
hence of $N_q$. Conditioning therefore preserves its uniform law.

For $N_q\sim\operatorname{Bin}(n,a/n)$, compare adjacent counts:
\[
 \frac{\Prb(N_q=r+1)}{\Prb(N_q=r)}
 =\frac{\binom n{r+1}}{\binom nr}\frac{q}{1-q}
 =\frac{n-r}{r+1}\frac{a}{n-a},
 \qquad 0\le r<n.
\]
The ratio exceeds one for $r<a$ (probabilities increase) and is
below one for $r\ge a$ (they decrease). Thus $a$ is the most likely
of $n+1$ counts, giving $p_a:=\Prb(N_q=a)\ge1/(n+1)$.
By comparison invariance and conditioning,
\[
 \E_\pi[e^{\theta C_a}]
 =\frac{\E_{\rm iid}[e^{\theta C(q)}\mathbf1[N_q=a]]}{p_a}
 \le(n+1)\E_{\rm iid}[e^{\theta C(q)}]
 \le2A_0n^{\gamma_0+1}.
\] 
\end{proof}

\subsubsection{Metric gaps and the upper tail}
\label{tr:metric-tail}

\begin{samepage}
\begin{lemma}[Convex combination of crossings]\label{tr:gaps}
Let the keys have a fixed strict order and associated nondecreasing real values
$x_{(1)}\le\cdots\le x_{(n)}$. If $\OPT(X)=x_{(n)}-x_{(1)}>0$, then every placement satisfies
\[
 \rho(A)=\sum_{a=1}^{n-1}w_a C_a,\qquad
 w_a=\frac{x_{(a+1)}-x_{(a)}}{\OPT(X)},\quad
 w_a\ge0,\quad\sum_a w_a=1.
\]
Consequently,
$\E[e^{\theta\rho(A)}]\le\sup_a\E[e^{\theta C_a}]$.
\end{lemma}
\end{samepage}

\phantomsection\label{proof:gaps}
\begin{proof}
An edge joining ranks $r<s$ spans exactly the gaps $r,\ldots,s-1$.
Summing over edges gives
\[
 \cost(A)=\sum_{a=1}^{n-1}(x_{(a+1)}-x_{(a)})C_a.
\]
Divide by $\OPT(X)$. Each weight is the fraction of the total
value range occupied by its gap, so the weights sum to one.
Convexity and expectation give
\[
 \E[e^{\theta\sum_aw_aC_a}]
 \le\sum_aw_a\E[e^{\theta C_a}]
 \le\sup_a\E[e^{\theta C_a}].
\]
No independence of crossing counts is needed.
\end{proof}

\phantomsection\label{proof:upper}
\begin{proof}[Proof of Theorem~\ref{thm:upper}]
For distinct values, Proposition~\ref{prob:fixed-cut-mgf} and
Lemmas~\ref{tr:mode} and~\ref{tr:gaps} give
\begin{equation}\label{eq:ratio-mgf}
 \E[e^{\theta\rho(A)}]\le2A_0n^{\gamma_0+1}.
\end{equation}

For ties, we verify that the priority order is independent of the
arrival permutation. Let $\pi$ uniformly permute occurrences and
let $\sigma$ be an independent uniform priority permutation
assigned by arrival time.
Occurrence $i$ receives priority
$\xi_i=\sigma_{\pi^{-1}(i)}$. For each fixed $\pi$, the entire vector
$\xi$ is uniform, hence independent of $\pi$.
Conditional on $\xi$, the lexicographic keys $(x_i,\xi_i)$ form a fixed
strict order arriving uniformly. The output stores original values.
In strict key order, these values are nondecreasing. The strict-key
argument and Lemma~\ref{tr:gaps} apply, with zero weights within ties.
Averaging over $\xi$ proves
\eqref{eq:ratio-mgf}. Constant inputs have ratio one and are covered
by enlarging $A_0$ if necessary.

The exponential Markov inequality now gives
\[
 \Prb[\rho(A)>t]\le2A_0n^{\gamma_0+1}e^{-\theta t}.
\]
Taking $a_0=(\gamma_0+1)/\theta$, $a_1=\log(2A_0)/\theta$, and
$t=a_0\log n+a_1+z$ proves the claimed exponential tail.
For $z=d\log n/\theta$, the failure probability is at most $n^{-d}$.
Since $\log n\ge\log2$, this gives the high probability form with
$K_d=a_0+d/\theta+a_1/\log2$.
Section~\ref{alg:implementation} supplies the implementation bound.
\end{proof}

\begin{samepage}
\begin{corollary}[Simultaneous rank-cut control]\label{cor:all-cuts}
For distinct fixed keys, there are absolute $c_{\mathrm{cut}},c_{\mathrm{add}},\theta>0$ such that
\[
 \Prb_\pi\!\left[
 \max_{1\le a<n}C_a>c_{\mathrm{cut}}\log n+c_{\mathrm{add}}+z
 \right]\le e^{-\theta z}\qquad(z\ge0).
\]
In particular, with high probability a single output rank arrangement has
$O(\log n)$ competitive ratio simultaneously for every nonconstant
nondecreasing assignment of real values to those ranks.
\end{corollary}
\end{samepage}

\begin{proof}
Exponential Markov inequality and a union bound give
\[
 \Prb_\pi[\max_aC_a>t]
 \le\sum_{a=1}^{n-1}e^{-\theta t}\E_\pi[e^{\theta C_a}]
 \le2A_0n^{\gamma_0+2}e^{-\theta t}.
\]
Take $c_{\mathrm{cut}}=(\gamma_0+2)/\theta$ and
$c_{\mathrm{add}}=\log(2A_0)/\theta$.
Every convex combination in Lemma~\ref{tr:gaps} is bounded by
$\max_a C_a$, proving the simultaneous guarantee.
\end{proof}

\subsection{Unknown i.i.d.\ input}
\label{tr:iid}

\phantomsection\label{proof:unknown-iid}
\begin{proof}[Proof of Corollary~\ref{cor:unknown-iid}]
Independently permute the complete i.i.d.\ sample. Conditional on the
sample, Theorem~\ref{thm:upper} applies with uniform constants.
Averaging preserves the failure bound, and permutation preserves the
i.i.d.\ law. The shuffle is only for analysis. The algorithm neither
shuffles its input nor needs to know its law.
No moment assumption is needed, since $\rho\le n-1$ for nonconstant
inputs and $\rho=1$ otherwise.
\end{proof}

Integrating the upper tail also gives $\E[\rho]=O(\log n)$ uniformly
over fixed multisets and unknown i.i.d.\ laws.

\section{Lower bounds and instance-sensitive guarantees}
\label{sec:lower-and-instances}

We first show that every randomized online algorithm incurs cost
$\Omega(\log n)$ with high probability under i.i.d.\ uniform input, then
extract a fixed hard input in random order. We then use numerical gaps
to sharpen our upper bounds for structured inputs.

\subsection{Worst-case lower bounds}
\label{lb:section}\label{sec:lower-overview}

We prove Theorems~\ref{thm:uniform-lower} and~\ref{thm:random-lower}
by adapting Hu's suffix-scale argument~\cite{Hu2026}.
With $i$ arrivals remaining, we show that, with high probability over
the remaining values, every completion has many long edges.
Bounds at doubling suffix sizes concern the same final array;
integrating them gives a lower bound on realized cost.

Let $X_1,\ldots,X_n$ be independent uniforms on $[0,1]$.
Fix an online algorithm $\mathcal A$ with random tape $\zeta$
independent of the input. Each placement is a measurable function of
$(n,\zeta,X_1,\ldots,X_t)$, uses an empty cell, and is irrevocable.
The algorithm may use the full numerical values.
For its completed array $A$, define
\[
 \operatorname{long}_A(u)
 :=\bigl|\{r\in\{1,\ldots,n-1\}:|A[r+1]-A[r]|>u\}\bigr|.
\]
Each edge is counted for thresholds below its length, so integration
recovers the total edge length:
\begin{equation}
 \cost(A)=\int_0^1\operatorname{long}_A(u)\,du.
 \label{lb:layer-cake}
\end{equation}

\subsubsection{Frozen anchors and suffix components}

Fix $1\le i\le n$ and freeze the algorithm after $n-i$ arrivals.
An occupied cell adjacent to an empty cell is an \emph{active anchor};
write $\mathsf{Anc}_i$ for their set. At most $i$ runs of empty cells,
each with at most two occupied boundary cells, give
\begin{equation}
 |\mathsf{Anc}_i|\le2i.
 \label{lb:anchor-count}
\end{equation}
The empty cells are precisely the positions that will hold the suffix.
Reindex the remaining inputs in arrival order as
$Y_j=X_{n-i+j}$, $1\le j\le i$.
Condition on $\cF_i^{\mathrm{fr}}=\sigma(\zeta,X_1,\ldots,X_{n-i})$:
the state and anchors are fixed, while the suffix values remain
independent uniforms.

For $\delta>0$, connect suffix occurrences $j,k$ if
$|Y_j-Y_k|\le\delta$; let $\operatorname{cc}_i(\delta)$ count the
components. Define
\[
 \operatorname{near}_i(\delta)
 :=\bigl|\{j:\exists h\in\mathsf{Anc}_i,
                 \ |Y_j-\operatorname{val}(h)|\le\delta\}\bigr|,
\]
where $\operatorname{val}(h)$ is the anchor's value.
This counts suffix occurrences, not anchors, and counts each only once.
Both counts depend only on the suffix values and frozen anchors,
not on future placements.

\begin{lemma}[Components force long edges]
\label{lb:component-edges}
For every frozen state, every suffix realization, and every legal completion,
\begin{equation}
 \operatorname{long}_A(\delta)
 \ge\operatorname{cc}_i(\delta)-\operatorname{near}_i(\delta)-1.
 \label{lb:component-inequality}
\end{equation}
\end{lemma}

\phantomsection\label{proof:component-edges}
\begin{proof}
Delete all final edges longer than $\delta$, leaving
$\operatorname{long}_A(\delta)+1$ blocks of short edges.
Within each block, call a maximal consecutive sequence of suffix
positions a \emph{suffix run}. Each run belongs to one proximity component.

The first suffix run need not follow a prefix cell within its block;
a block may consist entirely of suffix positions. Every later run starts
immediately after a prefix cell, which was next to an empty cell
when frozen and is therefore an anchor.
Their edge is short, so the run's first occurrence is counted by
$\operatorname{near}_i(\delta)$.
These occurrences are distinct. Since every component appears in a run,
\[
 \operatorname{cc}_i(\delta)
 \le \#\{\text{suffix runs}\}
 \le \operatorname{long}_A(\delta)+1+\operatorname{near}_i(\delta),
\]
proving the claim.
\end{proof}

\subsubsection{Concentration at one suffix size}

Set $\delta_i=1/(100i)$. We show that components are numerous and
anchor neighbors are few. For $0<p<a<1$, write the Bernoulli relative entropy
\[
 \operatorname{kl}(a\Vert p)
 :=a\ln\frac ap+(1-a)\ln\frac{1-a}{1-p}.
\]
In the following bound, $p$ bounds each conditional success probability,
while $a>p$ is the success fraction to be exceeded.

\begin{lemma}[Conditional Bernoulli bound]
\label{lb:bernoulli}
Fix $0<p<1$. Let nonnegative integer-valued random variables
$\omega_1,\ldots,\omega_i$ be adapted to sigma-fields
$\mathcal H_0\subseteq\cdots\subseteq\mathcal H_i$, with
$\Prb(\omega_j\ge1\mid\mathcal H_{j-1})\le p$ almost surely.
For $\lambda\ge0$,
\begin{equation}
 \E\left[
   e^{\lambda\sum_{j=1}^i\mathbf1[\omega_j\ge1]}
   \mid\mathcal H_0
 \right]
 \le(1-p+pe^\lambda)^i,
 \label{lb:bernoulli-mgf}
\end{equation}
and, for $p<a<1$,
\[
 \Prb\left[
   \sum_{j=1}^i\mathbf1[\omega_j\ge1]>ai
   \mid\mathcal H_0
 \right]
 \le e^{-i\operatorname{kl}(a\Vert p)}.
\]
\end{lemma}

\begin{proof}
For $\lambda\ge0$,
\[
 \E[e^{\lambda\mathbf1[\omega_j\ge1]}\mid\mathcal H_{j-1}]
 =1+(e^\lambda-1)\Prb(\omega_j\ge1\mid\mathcal H_{j-1})
 \le1-p+pe^\lambda.
\]
Iterating conditional expectations gives~\eqref{lb:bernoulli-mgf}.
Conditional exponential Markov inequality with
$e^\lambda=a(1-p)/(p(1-a))$ gives the tail.
\end{proof}

\begin{lemma}[Many components and few anchor neighbors]
\label{lb:single-statistics}
For every $1\le i\le n$, conditionally on $\cF_i^{\mathrm{fr}}$,
\begin{align}
 \Prb\left[\operatorname{cc}_i(\delta_i)<\frac{4i}{5}
                  \mid\cF_i^{\mathrm{fr}}\right]
 &\le e^{-c_{\mathrm{comp}}i},
 \label{lb:components-tail}\\
 \Prb\left[\operatorname{near}_i(\delta_i)>\frac{i}{10}
                  \mid\cF_i^{\mathrm{fr}}\right]
 &\le e^{-c_{\mathrm{anc}}i},
 \label{lb:anchors-tail}
\end{align}
where
\[
 c_{\mathrm{comp}}=\operatorname{kl}(1/10\Vert1/50),
 \qquad c_{\mathrm{anc}}=\operatorname{kl}(1/10\Vert1/25).
\]
\end{lemma}

\begin{proof}
Reveal $Y_1,\ldots,Y_i$ in order, with
$\mathcal H_j=\sigma(\cF_i^{\mathrm{fr}},Y_1,\ldots,Y_j)$.
Let $\omega_j$ count existing components, not points, within distance
$\delta_i$ of $Y_j$.
Always $\omega_j\le2$: three components would supply three representatives
in an interval of length $2\delta_i$, forcing two to be within
$\delta_i$ and hence already connected.
A new point creates a component and merges the $\omega_j$ components it
meets, changing the count by $1-\omega_j$. Thus
\begin{equation}
 \operatorname{cc}_i(\delta_i)
 =i-\sum_{j=1}^i\omega_j
 \ge i-2\sum_{j=1}^i\mathbf1[\omega_j\ge1].
 \label{lb:component-insertion}
\end{equation}
The $j-1$ previous points have neighborhoods of total length at most
$2\delta_i(j-1)=(j-1)/(50i)\le1/50$.
Since $Y_j$ is conditionally uniform, it follows that
$\Prb(\omega_j\ge1\mid\mathcal H_{j-1})\le1/50$.
If $\sum_j\mathbf1[\omega_j\ge1]\le i/10$, then
$\operatorname{cc}_i(\delta_i)\ge i-2(i/10)=4i/5$.
Consequently, Lemma~\ref{lb:bernoulli} with $p=1/50$, $a=1/10$,
and $\mathcal H_0=\cF_i^{\mathrm{fr}}$ proves~\eqref{lb:components-tail}.

For near occurrences, the union of the anchors' $\delta_i$-neighborhoods,
intersected with $[0,1]$, is a fixed region of length
$p_i\le2\delta_i|\mathsf{Anc}_i|\le1/25$.
Each independent suffix value falls there with probability $p_i$, so
$\operatorname{near}_i(\delta_i)\mid\cF_i^{\mathrm{fr}}
\sim\operatorname{Bin}(i,p_i)$.
Applying Lemma~\ref{lb:bernoulli} to the indicators of these memberships,
with $p=1/25$ and $a=1/10$, proves~\eqref{lb:anchors-tail}.
\end{proof}

Use the smaller exponent
$c_\circ=\min\{c_{\mathrm{comp}},c_{\mathrm{anc}}\}>0$
to combine these failures.
The value and first derivative of $\operatorname{kl}(a\Vert p)$
vanish at $a=p$, while its second derivative is $1/(a(1-a))\ge4$.
Hence $\operatorname{kl}(a\Vert p)\ge2(a-p)^2$ for $a\ge p$, giving
\begin{equation}
 c_\circ\ge\frac9{1250}.
 \label{lb:positive-exponent}
\end{equation}

\begin{lemma}[Single-scale lower tail]
\label{lb:single-scale}
For every integer $5\le i\le n$,
\begin{equation}
 \Prb\left[\operatorname{long}_A(\delta_i)<\frac i2
                 \mid\cF_i^{\mathrm{fr}}\right]
 \le2e^{-c_\circ i}.
 \label{lb:single-scale-tail}
\end{equation}
\end{lemma}

\begin{proof}
When neither event in Lemma~\ref{lb:single-statistics} occurs,
Lemma~\ref{lb:component-edges} gives, for every completion,
\[
 \operatorname{long}_A(\delta_i)
 \ge\frac{4i}{5}-\frac{i}{10}-1\ge\frac i2
 \qquad(i\ge5).
\]
A conditional union bound gives failure probability at most
$e^{-c_{\mathrm{comp}}i}+e^{-c_{\mathrm{anc}}i}\le2e^{-c_\circ i}$.
\end{proof}

\subsubsection{Integration across doubling suffix sizes}

\phantomsection\label{proof:uniform-lower}
\begin{proof}[Proof of Theorem~\ref{thm:uniform-lower}]
The doubling suffix sizes below select analysis times in one execution.
We combine their bounds over threshold intervals with disjoint interiors
in the same cost integral.
Starting near $\sqrt n$ gives exponentially small failure probability
while retaining $\Theta(\log n)$ scales.
For $n\ge2^{16}$, set
\[
 i_0=2^{\lceil(\log_2n)/2\rceil},\qquad i_j=2^ji_0,
 \qquad J_*=\max\{j\ge0:i_j\le n/4\}.
\]
Since $\sqrt n\le i_0<2\sqrt n$,
\begin{equation}
 J_*=\left\lfloor\log_2\frac{n}{4i_0}\right\rfloor
 \ge\frac12\log_2n-4.
 \label{lb:scale-count}
\end{equation}
Averaging out the frozen state in Lemma~\ref{lb:single-scale} gives
$\Prb[\operatorname{long}_A(\delta_{i_j})<i_j/2]\le2e^{-c_\circ i_j}$.
These concern the same final array $A$; no independence is needed:
\begin{align}
 \Prb\!\left[\exists j<J_*:
       \operatorname{long}_A(\delta_{i_j})<\frac{i_j}{2}\right]
 &\le2\sum_{j=0}^{J_*-1}e^{-c_\circ i_j}
 \le4e^{-c_\circ\sqrt n}.
 \label{lb:dyadic-failure}
\end{align}
Indeed, $2^j\ge j+1$ bounds the sum by
$2e^{-c_\circ i_0}/(1-e^{-c_\circ i_0})\le4e^{-c_\circ i_0}$;
$i_0\ge256$ and~\eqref{lb:positive-exponent} ensure $c_\circ i_0>\ln2$.

On the common good event, $\operatorname{long}_A(u)\ge i_j/2$
for $u\in[\delta_{i_{j+1}},\delta_{i_j}]$, since the count is
nonincreasing. Doubling $i_j$ halves $\delta_{i_j}$:
each interval has width $1/(200i_j)$ and contributes at least
$(i_j/2)/(200i_j)=1/400$.
Their interiors are disjoint, so
\[
 \cost(A)
 \ge\sum_{j=0}^{J_*-1}\int_{\delta_{i_{j+1}}}^{\delta_{i_j}}
              \operatorname{long}_A(u)\,du
 \ge\frac{J_*}{400}\ge\frac{\log_2n}{1600},
\]
by~\eqref{lb:layer-cake}, \eqref{lb:scale-count}, and $\log_2n\ge16$.
Thus
\begin{equation}
 \Prb\!\left[\cost(A)<\frac{\log_2n}{1600}\right]
 \le4e^{-c_\circ\sqrt n}.
 \label{lb:uniform-explicit}
\end{equation}
The theorem follows with $c_0=1/(1600\ln2)$, $c_1=c_\circ$,
and $n_0=2^{16}$.
These constants are algorithm-independent; the probability includes
the random tape.
\end{proof}

\subsubsection{A hard fixed multiset in random order}

\phantomsection\label{proof:random-lower}
\begin{proof}[Proof of Theorem~\ref{thm:random-lower}]
Fix $n\ge n_0$ and $\mathcal A$, including its tape distribution.
Let $V=(V_1,\ldots,V_n)$ consist of independent uniforms on $[0,1]$.
Independently draw a uniform permutation $\pi$ of $[n]$ and the tape
$\zeta$. Then $X_t=V_{\pi(t)}$ has the same i.i.d.\ uniform law.
Conditional on $V=v$, the values are fixed; only their order and
the algorithm's randomness remain.

For a vector $v$ with distinct coordinates, define its lower-bound
failure probability
\[
 g(v):=\Prb_{\pi,\zeta}
   [\rho(\mathcal A;v_{\pi(1)},\ldots,v_{\pi(n)},\zeta)<c_0\ln n].
\]
Extend $g$ by zero to repeated-coordinate vectors; it is measurable
and bounded. Almost surely $V$ has distinct coordinates and
$0<\OPT(V)\le1$, so $\rho(A)\ge\cost(A)$.
Here we condition on the entire value vector, not the frozen prefix
and tape. Total expectation and Theorem~\ref{thm:uniform-lower} give
\[
 \begin{aligned}
 \E_V[g(V)]
 &=\Prb_{V,\pi,\zeta}[\rho(A)<c_0\ln n]\\
 &\le\Prb_{V,\pi,\zeta}[\cost(A)<c_0\ln n]
 \le4e^{-c_1\sqrt n}.
 \end{aligned}
\]
Some fixed vector $v\in[0,1]^n$ with distinct coordinates has failure
probability no larger than this average.
Its multiset $X=\{v_1,\ldots,v_n\}$ proves the claim.
It may depend on $n$ and $\mathcal A$, but not on the realized
permutation or tape: both were already averaged in $g$.
\end{proof}

\subsection{Instance-sensitive guarantees}\label{sec:instance-sensitive}

We now return to \EmpiricalSort{} and its fixed-cut upper bounds. A rank
cut near either end of the sorted input can be crossed only a few times.
Combining this deterministic restriction with the exponential-moment bound
gives sharper guarantees for inputs whose large gaps lie near the extremes.
The algorithm is unchanged and is not given the input's structure.

Fix a nonconstant multiset $X$ of $n\ge2$ real values, sorted as
$x_{(1)}\le\cdots\le x_{(n)}$. Its range is
$\OPT(X)=x_{(n)}-x_{(1)}>0$, and its normalized gaps are
$w_a=(x_{(a+1)}-x_{(a)})/\OPT(X)$.
Let $C_a$ count output edges with one rank at most $a$ and the other
larger than $a$. Lemma~\ref{tr:gaps} gives $\rho(A)=\sum_a w_aC_a$.
Each array vertex has degree at most two, so every placement satisfies
$C_a\le\min\{2a,2(n-a),n-1\}$. Define the \emph{gap profile}
\begin{equation}\label{tr:gap-profile}
 u_a:=\min\{2a,2(n-a),n-1\},\qquad
 \Psi_X(\ell):=\sum_{a=1}^{n-1}w_a\min\{u_a,\ell\}\quad(\ell\ge0).
\end{equation}
The profile depends only on $X$ and is at most $\min\{\ell,n-1\}$.
In particular, every placement satisfies
$\rho(A)\le\Psi_X(C_{\max})\le C_{\max}$, where
$C_{\max}=\max_a C_a$.

\subsubsection{The gap profile and its tail bound}

Using the constants from Proposition~\ref{prob:fixed-cut-mgf}, set
\begin{equation}\label{tr:profile-threshold}
 \ell_n:=\frac{\log(2A_0)+(\gamma_0+1)\log n}{\theta},\qquad
 \ell_n(\delta):=\ell_n+\frac{\log(1/\delta)}{\theta}\quad(0<\delta<1).
\end{equation}

\begin{proposition}[Instance-sensitive expectation and tail]\label{tr:profile}
For every $n\ge2$, every fixed nonconstant real multiset $X$ of size $n$,
and every $0<\delta<1$, the output of $\EmpiricalSort$ satisfies
\begin{equation}\label{tr:profile-bounds}
 \E_{\pi,\sigma}[\rho(A)]\le\Psi_X(\ell_n),\qquad
 \Prb_{\pi,\sigma}\!\left[\rho(A)>\Psi_X(\ell_n(\delta))\right]\le\delta.
\end{equation}
The constants in~\eqref{tr:profile-threshold} are independent of $X$ and $n$.
\end{proposition}

\begin{proof}
Lemma~\ref{tr:mode} gives $\E[e^{\theta C_a}]\le e^{\theta\ell_n}$.
For repeated values, condition on the occurrence priorities as in the proof
of Theorem~\ref{thm:upper}, then average. The priorities are independent of
the arrival permutation, and the gap weights and profile are unchanged by
the order within ties. Jensen's inequality gives $\E[C_a]\le\ell_n$.
Together with $C_a\le u_a$ and $\rho(A)=\sum_a w_aC_a$, this proves
the expectation bound.

{For the tail, fix $\ell\ge0$.
We bound the contributions of cuts with $u_a\le\ell$
deterministically and control the remaining contribution through
the exponential moment of a normalized weighted average.
Write}
\[
 \mathcal I_\ell=\{a:u_a>\ell\},\qquad
 \alpha_\ell=\sum_{a\in\mathcal I_\ell}w_a,\qquad
 D_{\mathrm{det}}(\ell)=\sum_{a\notin\mathcal I_\ell}w_au_a.
\]
If $\alpha_\ell=0$, then $\rho(A)\le\Psi_X(\ell)$ deterministically.
Otherwise, put $Z_\ell=\sum_{a\in\mathcal I_\ell}(w_a/\alpha_\ell)C_a$.
The deterministic cut bounds and convexity give
\[
 \rho(A)\le D_{\mathrm{det}}(\ell)+\alpha_\ell Z_\ell,
 \qquad
 \E[e^{\theta Z_\ell}]
 \le\sum_{a\in\mathcal I_\ell}\frac{w_a}{\alpha_\ell}
                \E[e^{\theta C_a}]
 \le e^{\theta\ell_n}.
\]
Since $\Psi_X(\ell)=D_{\mathrm{det}}(\ell)+\alpha_\ell\ell$,
exponential Markov inequality yields
\[
 \Prb[\rho(A)>\Psi_X(\ell)]
 \le\Prb[Z_\ell>\ell]\le e^{\theta(\ell_n-\ell)}.
\]
Take $\ell=\ell_n(\delta)$.
\end{proof}

Constant multisets have zero cost and ratio one by
\eqref{eq:ratio-definition}.

\subsubsection{A narrow core and few outliers}

\begin{corollary}[A narrow core and few outliers]\label{tr:narrow-core}
Suppose $\OPT(X)>0$, $0\le r\le n-1$ is an integer, and at least $n-r$ input
values lie in an interval of width at most $\varepsilon \OPT(X)$, where
$0\le\varepsilon\le1$. For every $0<\delta<1$,
\begin{equation}\label{tr:core-bounds}
 \E[\rho(A)]\le2r+\varepsilon\ell_n,\qquad
 \Prb\!\left[\rho(A)>2r+\varepsilon\ell_n(\delta)\right]\le\delta.
\end{equation}
In particular, for every fixed $d>0$ the competitive ratio is
$O_d(r+\varepsilon\log n)$ with failure probability at most $n^{-d}$.
The interval and the parameters $r,\varepsilon$ need not be known to the
algorithm.
\end{corollary}

\begin{proof}
Let $p,q$ be the first and last ranks in the stated interval. Then
$(p-1)+(n-q)\le r$ and
$\sum_{a=p}^{q-1}w_a=(x_{(q)}-x_{(p)})/\OPT(X)\le\varepsilon$.
Outside these internal gaps, $u_a\le2r$: use $u_a\le2a$ for $a<p$ and
$u_a\le2(n-a)$ for $a\ge q$. Hence, for every $\ell\ge0$,
\[
 \Psi_X(\ell)
 \le\ell\sum_{a=p}^{q-1}w_a
       +2r\sum_{a\notin\{p,\ldots,q-1\}}w_a
 \le\varepsilon\ell+2r.
\]
Apply Proposition~\ref{tr:profile} at $\ell_n$ and $\ell_n(\delta)$.
\end{proof}

For example, take $n\ge3$, one value equal to $1$, and the other $n-1$
values in $[0,1/\log n]$, including $0$. With $r=1$ and
$\varepsilon=1/\log n$, the corollary gives $O(1)$ expected ratio and
$O_d(1)$ ratio with failure probability at most $n^{-d}$, even when all
values are distinct. The same conclusion holds whenever $r=O(1)$ and
$\varepsilon=O(1/\log n)$.

\subsubsection{Shrinking central spans}

\begin{corollary}[Shrinking central spans]\label{tr:central-span}
For a fixed nonconstant multiset define
\[
 \tau_j:=\frac{x_{(n-j)}-x_{(j+1)}}{\OPT(X)},
 \qquad 0\le j<\lfloor n/2\rfloor.
\]
Suppose $\tau_j\le \Lambda/(j+1)$ throughout this range, for some $\Lambda\ge1$.
For every $0<\delta<1$,
\begin{align}
 \E[\rho(A)]&\le2\Lambda\bigl(1+\log(1+\ell_n)\bigr),\label{tr:central-mean}\\
 \Prb\!\left[\rho(A)>2\Lambda\bigl(1+\log(1+\ell_n(\delta))\bigr)\right]
 &\le\delta.\label{tr:central-tail}
\end{align}
For families with a uniform constant $\Lambda$, these are
$O(1+\log\log n)$ in expectation and
$O_d(1+\log\log n)$ with failure probability at most $n^{-d}$.
\end{corollary}

\begin{proof}
For $1\le k\le\lfloor n/2\rfloor$,
\[
 \Psi_X(2k)
 \le2\sum_{a=1}^{n-1}w_a\min\{a,n-a,k\}
 =2\sum_{j=0}^{k-1}\sum_{a=j+1}^{n-j-1}w_a
 =2\sum_{j=0}^{k-1}\tau_j.
\]
The first equality counts the integers $0\le j<\min\{a,n-a,k\}$; the inner
gap sum then telescopes to $\tau_j$.
For $\ell>0$, take $k=\min\{\lceil\ell/2\rceil,\lfloor n/2\rfloor\}$.
Either $2k\ge\ell$, or every $u_a\le2k$; in both cases
$\Psi_X(\ell)\le\Psi_X(2k)$. Therefore
\[
 \Psi_X(\ell)\le2\Lambda\sum_{j=1}^k\frac1j
 \le2\Lambda(1+\log k)\le2\Lambda\bigl(1+\log(1+\ell)\bigr).
\]
For $\ell=0$ the bound is immediate. Proposition~\ref{tr:profile}
completes the proof.
\end{proof}

For the reciprocal multiset $X=\{1,1/2,\ldots,1/n\}$,
$x_{(a)}=1/(n-a+1)$ and $\OPT(X)=1-1/n$, so
\[
 \tau_j=\frac{1/(j+1)-1/(n-j)}{1-1/n}\le\frac2{j+1}.
\]
Corollary~\ref{tr:central-span} gives $O(1+\log\log n)$ expected ratio
and $O_d(1+\log\log n)$ ratio with failure probability at most $n^{-d}$.

For i.i.d.\ inputs, condition on the sampled multiset. Whenever it satisfies
one of the stated structural conditions, the corresponding bound applies.
The probability that the condition fails is added to the conditional
failure bound.

\section*{AI methodologies}
\phantomsection\label{ai-usage}
\addcontentsline{toc}{section}{AI methodologies}
We used ChatGPT (GPT-6 Astra and GPT-6 Astra Pro) to explore some of
the ideas developed in this paper, including reverse conditional
moment bounds for growing prefixes, the accounting of persistent
partition errors across merging phases, and the transfer from
fixed threshold crossings to arbitrary fixed inputs. We also used
it to develop, check, and refine proofs, search
the literature, and edit the manuscript and its \LaTeX{} source.
The authors are responsible for all claims, proofs, references,
and other content of the paper.

\clearpage
\phantomsection\label{references-start}
\addcontentsline{toc}{section}{References}
\bibliography{references}
\end{document}